\documentclass[12pt]{article}
\usepackage{csquotes}
\usepackage{natbib}
\usepackage[utf8]{inputenc} % allow utf-8 input
\usepackage{url}            % simple URL typesetting
\usepackage{booktabs}       % professional-quality tables
\usepackage{amsfonts}       % blackboard math symbols
\usepackage{nicefrac}       % compact symbols for 1/2, etc.
\usepackage{microtype}      % microtypography
\usepackage{algorithm}
\usepackage{algorithmic}
\usepackage{grffile}
\usepackage{amsmath,amsfonts,amsthm} % Math packages
\usepackage{amssymb}
\usepackage{graphicx,psfrag,epsf}
\usepackage{listings}
\usepackage{enumerate}
\usepackage{color}

\newtheorem{lemma}{Lemma}
\usepackage{lscape}
\usepackage{bm}
\usepackage{rotating}
\usepackage{lipsum,multicol}
\usepackage{comment}
\usepackage{amssymb}
\usepackage{multirow}
\usepackage{xr-hyper}
\usepackage[colorlinks = true,
           linkcolor = blue,
           urlcolor  = magenta,
           citecolor = blue,
           anchorcolor = blue]{hyperref}
\allowdisplaybreaks[4]

\newcommand{\vertiii}[1]{{\left\vert\kern-0.25ex\left\vert\kern-0.25ex\left\vert #1
    \right\vert\kern-0.25ex\right\vert\kern-0.25ex\right\vert}}
\newtheorem{theorem}{Theorem}
\newtheorem{remark}{Remark}
\newcommand{\blind}{1}

\begin{document}

\def\spacingset#1{\renewcommand{\baselinestretch}%
{#1}\small\normalsize} \spacingset{1}

%%%%%%%%%%%%%%%%%%%%%%%%%%%%%%%%%%%%%%%%%%%%%%%%%%%%%%%%%%%%%%%%%%%%%%%%%%%%%%
%Bagus\footnote{Bagus means good and nice in Indonesian and Malay}:
%\spacingset{1.48}
\if1\blind
{
  \title{\bf  Subgroup Analysis for Longitudinal data via Semiparametric Additive Mixed Effect Model}
  \author{Xiaolin Bo, %\thanks{
    %The authors gratefully acknowledge \textit{please remember to list all relevant funding sources in the unblinded version}}\hspace{.2cm}
    Weiping Zhang %\thanks{
    %The authors gratefully acknowledge \textit{please remember to list all relevant funding sources in the unblinded version}}\hspace{.2cm}
    \\%\thanks{
    %The authors gratefully acknowledge \textit{please remember to list all relevant funding sources in the unblinded version}}\hspace{.2cm}
   %\\
    Department of Statistics and Finance, \\University of Science and Technology of China\\
    \date{}
}
  \maketitle
} \fi

\if0\blind
{
  \bigskip
  \bigskip
  \bigskip
  \begin{center}
    {\LARGE\bf Title}
\end{center}
  \medskip
} \fi

\bigskip

\begin{abstract}
In this paper, we propose a general subgroup analysis framework based on semiparametric additive mixed effect models in longitudinal analysis, which can identify subgroups on each covariate and  estimate the corresponding regression functions simultaneously. In addition, the proposed procedure is applicable for both balanced and unbalanced longitudinal data.
A backfitting combined with $k$-means algorithm is developed to estimate each semiparametric additive component across subgroups and detect subgroup structure on each covariate respectively. The actual number of groups is estimated by minimizing a Bayesian information criteria. The numerical studies demonstrate the efficacy and accuracy of the proposed procedure in identifying the subgroups and estimating the regression functions.  In addition, we illustrate the usefulness of our method with an application to PBC data and provide a meaningful partition of the population.

\end{abstract}

\noindent%
   \vfill

%\newpage

\section{Introduction}

Subgroup analysis has emerged as important drug development tool with the demand of precision medicine emerging and rising %,, and we are aware of the importance of personalized medical treatment,since the traditional ``one size fits all" approach may be not suitable for all individuals
\citep{foster2011subgroup,Schwalbeetal}. Consequently, there is an increasing need to distinguish homogeneous subgroups of individuals, detect the diverse patterns in the subpopulations, model the relationships between the response variable and predictors differently across the subpopulations and make the best personalized predictions for individuals belonging to different subgroups. Thus, virous statistical methods have been developed for subgroup identification in longitudinal data, such as decision trees, mixture models, regularization methods and change point methods etc.

Since the seminal book on classification and regression trees (CART) by \citep{breiman1984classification}, tree-based methods become  widely used for subgroup identification. In general, a tree recursively partitions the subjects into binary nodes until specific stopping rule is met and in this way subgroups are yielded. \cite{sela2012re} developed the RE-EM Tree procedure, fitting a mixed effect model by regarding the fixed effect as a regression tree and iteratively,  estimating the random effect and the fixed effect like EM algorithm rather than traditional maximum likelihood estimation. \cite{loh2013regression} extended GUIDE algorithm to longitudinal and multiresponse data, by means of treating each longitudinal data series as a curve and using chi-squared tests of the residual curve patterns to select a variable to split each node of the tree. Model-based recursive partitioning method developed by \cite{zeileis2008model} and \cite{seibold2016model}, fits a parametric model in each node, with splitting variable chosen by independence tests and parameter values estimated as solutions to the score equations, which is the partial derivatives of the log-likelihood. \cite{wei2020precision} proposed interaction tree for longitudinal trajectories, combining mixed effect models with regression splines to model the nonlinear progression patterns among repeated measures, and identify subgroups with differential treatment effects for two-sample comparisons in longitudinal randomized clinical trials.
% \cite{wei2020precision} propose interaction tree for longitudinal trajectories (IT-LT), combining mixed effect models with regression splines to model the nonlinear progression patterns among repeated measures, and identify subgroups with differential treatment effects for two-sample comparisons in longitudinal randomized clinical trials.

Furthermore, the growth mixture modeling  methods \citep{fraley2002model,song2007clustering,jung2008introduction}, have been widely utilized to identify and predict latent subpopulations for longitudinal data. For example, \cite{mcnicholas2010model} and \cite{mcnicholas2016model} developed a model-based clustering method (called \textit{longclust}) method for balanced longitudinal data.  \cite{shenqu2020} proposed a structured mixed-effects approach for longitudinal data to model subgroup distribution and identify subgroup membership simultaneously. In general, such approaches require to know the underlying distribution of data  and  the number of mixture components in advance. Comparably, clustering regression curves can be done to find subgroups. \cite{abraham2003unsupervised} proposed a clustering procedure which consists of two stages: fitting the functional data by B-splines and partitioning the estimated model coefficients using a $k$-means algorithm. They also shown their procedure possessing strong consistency. \cite{ma2006data} and \cite{coffey2014clustering} adopted the smoothing spline and penalized spline approximations under the mixed effect framework respectively to model time-course gene expression data and detect subpopulations. In addition, some distance-based clustering methods have been proposed to cluster the trajectories of longitudinal data, for instance, \cite{genolini2010kml} combined generalized Fréchet distance with k-means to achieve this goal. \cite{lv2020nonparametric} proposed a two-step classification algorithm which compares the $L_2$-distances between kernel estimates of nonparametric functions to estimate parameters of group memberships and the number of subgroups simultaneously. \cite{ZHANG201954} proposed a  quantile-regression-based clustering method for panel data by using a similar idea of k-means clustering to identify subgroups with heterogeneous slopes at a single quantile level or across multiple quantiles.

%However, some of these methods are limited to balanced data, some did not take correlations from the same individuals into account, and some require knowledge of the true number of subgroups.

There has also been a line of work on regularization methods. For example, \cite{ma2017concave} proposed a concave pairwise fusion learning method to identify subgroups whose heterogeneity is driven by unobserved latent factors and thus can be represented by subject-specific intercepts. \cite{zhang2019robust} employed penalized median regression to detect subgroups automatically and achieve robustness against outliers and heteroscedasticity in random errors.
\cite{LU2021104691} proposed a subgroup identification method based on concave fusion penalization and median regression for longitudinal data with dropouts. However, these aforementioned methods mainly focused on dividing the individuals into several groups according to the intercept or the whole list of all regression coefficients, not on detecting subgroups on each covariate separately. For this purpose, \cite{li2019subgroup} proposed  an estimation procedure combining the likelihood method and the change point detection with the binary segmentation algorithm under partial linear models.

There is a clear need to relax the parametric assumption
posed in \cite{li2019subgroup} as model misspecification may result in biased estimation. An attractive approach is the semiparametric additive mixed effect model, which retains the flexibility of the nonparametric model
but avoids the curse of dimensionality of a fully nonparametric model.
%and some of them are limited to balanced longitudinal data. Secondly, we are not only concerned with identifying subgroups, but also revealing the different functional relationships for each covariate of interest.
In this paper, we propose a very general framework to recognize the subgroup structure on each covariate and estimate the regression functions in each subgroup simultaneously, based on semiparametric additive mixed effect model. Specifically, using the densely observed data for each individual, we give a initial estimates of the parametric part and additive components in the model, pretending that there is no subgroups in the population. Then we adopt the backfitting and $k$-means algorithm to estimate each semiparametric additive component across subgroups and detect subgroup structure on each covariate. The utilization of mixed effect enables us capture the within-subject correlation among longitudinal measurements, while additive nonparametric components are helpful for us to characterize the nonlinear relationships between covariates and the response. %This is why we employ the semiparametric additive mixed effect model.

The major contributions in this paper can be outlined as follows. First, we  propose a very general framework for identifying subgroups on multiple covariates, which possesses the flexibility and interpretability of semiparametric additive mixed effect model. In addition, the proposed method can detect subgroups on each covariate, consequently it could make us more clear about which covariate contributes to the existence of subgroups among population. Second, the proposed model is applicable for both balanced and unbalanced longitudinal data.  Third, the proposed procedure holds some theoretical properties and computationally simplicity,  our simulation studies  also indicate the fine efficiency of our approach.

The rest of this paper is organized as follows. Section 2 elaborates the proposed subgroup identification methods. Section 3 discusses the asymptotic properties. Section 4 conducts simulation studies to evaluate  the proposed approach. Application to the PBC dataset is presented in Section 5. Finally, Section 6 concludes this paper by summarizing the main findings and outlining future research.

\section{Methods}

\subsection{The model}
Consider the following longitudinal dataset $(Y_{it}, {X}_{it}, {Z}_{it}, {S}_{i}),i=1,\ldots,n,t=1,\ldots,n_i$ collected from $n$ independent individuals, where $Y_{it}$ is the response variable for the $i$th subject at the $t$th follow-up of total $n_i$ measurements, $ { X}_{it}=(X_{it1},\ldots,X_{itp})^{\prime}$ is a $p$-dimensional vector of predictors associated with fixed effects, $ Z_{it}$ is a $q \times 1$ vector of covariates associated with random effects, $ S_{i}$ is a $r \times 1$ vector of baseline covariates.  We assume the following semiparametric additive mixed effect model with certain subgroup structure in the population:
\begin{equation}
    \label{1}
    Y_{it}=  S_{i}^{\prime} \beta +f_1(X_{it1}) + \ldots +f_p(X_{itp}) +  Z_{it}^{\prime} b_i + \varepsilon_{it}, \quad i = 1,\ldots,n, \quad t=1,\ldots, n_i,
\end{equation}
where
\begin{equation}
    \label{2}
     f_j(X_{itj})=
\begin{cases}
	f_{1j}(X_{itj}), & \text { when } i \in \Omega_{1, j} \\
    f_{2j}(X_{itj}), & \text { when } i \in \Omega_{2, j} \\
    \qquad \vdots & \qquad \vdots \\
    f_{m_j j}(X_{itj}), & \text { when } i \in \Omega_{m_j,j}.
\end{cases}
\end{equation}
with $\{\Omega_{kj}: 1 \leq k \leq m_j \}, j=1,\ldots,p$ representing an unknown partition of the subject index set $\{i: 1 \leq i \leq n\}$ on covariate $X_j$. Note that the number of subgroups $m_j$ is unknown either. 
The traditional random-effects model assumes that the random effects follow a certain distribution, for instance, a normal distribution, and focuses on the variance component estimation of the random effects. However, we do not impose any distribution assumption on $b_i$, but instead assume that the random effects have mean zero and variance $\sigma_b^2 >0$. In addition, $\varepsilon_{it}$ is the random error with zero mean and variance $\sigma_{\varepsilon}^2$.
%For convenience, we assume that the random effect $b_i$ follows $N(0, D(\theta))$, where $\theta$ is a $c \times 1$ vector of variance components, and the random error $\varepsilon_{it}$ follows $N(0,\sigma^2)$. 
Under this model, the trajectory of the $i$th subject over time is represented by the linear regression part, the group-specific unknown additive functions and the subject-specific random effect. Without loss of generality, we assume that all $E[f_j(\cdot)]=0$ to prevent identifiability problems.

%In most longitudinal studies, what we desire to know is the impact of the covariate $X_{it}$ on the response $Y_{it}$ for different subgroups, so we will focus on detecting subgroups on each covariate $X_j$. The aim of this paper is to estimate the additive functions for each group and subgroup subjects simultaneously. Moreover, the purpose of including parametric part covariates into the proposed model is that in practice, some baseline variables such as gender, treatment have very essential interpretability to response variable.
%The semiparametric additive mixed effect models have shown great flexibility. One of the nice things about linear models is that they are fairly straightforward to interpret.%, if you want to know how the prediction changes as you change $x_j$, you just need to know $\beta_j$.
%While additive models keep a lot of the nice properties of linear models, but are more flexible and the partial response function $f_j$ plays the same role as linear predictors in an additive model. Moreover, we also involve the linear parts and the random effect, which makes our model including both the fixed effect and By combining various advantages of linear models, nonparametric models and mixed effect models, we will have a more comprehensive description of longitudinal data.

Note that model $\eqref{1}$ is very flexible, it retains  the flexibility of the nonparametric model
but avoids the need to model a fully nonparametric model. The linear part possesses the ease of interpretability on the baseline covariates,  the random effect term represents the heterogeneity between individuals, and the additive nonparametric part reflect the nonlinear relationships between response and each covariate. The proposed model $\eqref{1}$ also allows subgroups exist on each covariate, thus for the $j$th predictor, we may have $m_j$ different nonparametric additive components and thus form $m_j$ different subgroups. The subjects with the same $f_j$ share a similar nonparametric dependence between the response and the $j$th predictor. It can be seen that the detected subgroups for different $j$ may have completer, partial, or zero overlap, so there may have at least $\text{max}_j m_j$ subgroups and at most $\sum_{j=1}^{p} m_j$ subgroups, where each subgroup has a distinct set of additive components. In the following, we use $|A|$ to denote the cardinality of a set $A$.

\subsection{Subgroup identification algorithm}

\subsubsection{Initial estimates}
In this section, we describe the procedure for initial estimates of the additive components. Let $\bm Y_i=(Y_{i1},\ldots,Y_{in_i})^{\prime}$, $\bm Z_{i}=(Z_{i1},\ldots,Z_{in_i})^{\prime}, \bm \varepsilon_i=(\varepsilon_{i1},\ldots, \varepsilon_{in_i})^{\prime}$, and let $\bm S_i = \mathbf{1}_{\boldsymbol{n}_{i}} S_i^{\prime}$, where $\mathbf{1}_{\boldsymbol{n}_{i}}$ is a $n_i \times 1$ vector with entries equal to one. Based on B-spline approximation, for each subject $i$, model $\eqref{1}$ can be written in matrix notation as
\begin{equation}
\label{3}
  \bm Y_i \approx  \bm S_i \beta + \bm B_{i1} \gamma_{i1}+ \ldots + \bm B_{ip} \gamma_{ip}+ \bm Z_{i} b_i + \bm \varepsilon_i, \quad i=1,\ldots,n,
\end{equation}
where $\bm B_{ij}=(B_{i1j}, \ldots, B_{in_ij})^{\prime}$ are B-spline basis functions, for $j=1,\ldots,p$,
\begin{equation}
    \label{4}
    \gamma_{ij}=
 \begin{cases}
    \gamma_{1j}, & \text { when } i \in \Omega_{1, j} \\
 	\gamma_{2j}, & \text { when } i \in \Omega_{2, j} \\
 	\quad \vdots & \qquad \vdots \\
 	\gamma_{m_j j}, & \text { when } i \in \Omega_{m_j,j}.
 \end{cases}
\end{equation}
In addition, let $\bm \Sigma_i$ and $\bm V_i$ denote the true and assumed working covariance of $\bm Y_i$, where $\bm V_i = \bm A_i^{1/2} \bm R_i \bm A_i^{1/2}$. Here, $\bm A_i$ is a $n_i \times n_i$ diagonal matrix of the marginal variance of $\bm Y_i$, and $\bm R_i$ is the corresponding working correlation matrix. In this situation, $\bm V_i$ is assumed to depend on a nuisance finite dimensional parameter vector $\tilde {\bm \theta}$.

In order to identify the subgroups on each covariate, we first pretend that all individuals come from the same group, or that there are no subgroups in the population, and obtain an initial estimate by fitting model $\eqref{3}$. The initial estimates can be obtained by approximating the additive components $f_1(\cdot),\ldots,f_p(\cdot)$ through a spline basis expansion and then employing an extension of the standard GEE. Since smoothing spline has higher computational cost, here we implement B-spline, which not only maintains comparable performance in estimating, but also reduces computation complexity.

The initial estimates of nonparametric functions, $\widehat f_1^{(0)},\ldots,\widehat f_p^{(0)}$, or to be more specific, the initial estimates of B-spline coefficients $\widehat \gamma_1^{(0)},\ldots,\widehat \gamma_p^{(0)}$ will be input into the subgroup identification procedure in the next subsection. \cite{huang2007efficient} has shown that the estimate of B-spline coefficient is efficient if the covariance structure is correctly specified and it is still consistent and asymptotically normal even if the covariance structure is misspecified.  The explicit derivation and proof can refer to \cite{huang2007efficient}, so we will not discuss it in detail here.

\subsubsection{Backfitting and subgroup pursuit}
Given model $\eqref{1}$, we can employ backfitting algorithm to fit additive models by iteratively solving
\begin{equation}
    \label{5}
    E[Y- S \beta - \sum_{k \neq j} f_k (X_k)| X_j] = f_j(X_j),j=1,\ldots,p,
\end{equation}
and at each stage we replace the conditional expectation of the partial residuals with a univariate smoother. We first set
\begin{equation}
    \label{6}
    W_{itj} = Y_{it} -  S_{i} \widehat \beta - \sum_{k \neq j} \widehat f_k^{(0)}(X_{itk}), j =1,\ldots,p,
\end{equation}
thus the problem can be transformed to univariate nonparametric regression. Various methods have been proposed to solve univariate nonparametric regression issue, for instance, kernel method, local polynomial regression and splines, and we are going to utilize the B-spline in each stage, since it has nice properties of efficiency and flexibility \citep{de1978practical,lorentz1993constructive}.

For $j =1,\ldots,p$, we aim to divide $n$ subjects into $m_j$ groups such that subjects with homogeneity being partitioned into the same group.
We first assume that each subject has its own regression function between $X_j$ and $W_j$. In other words, hypothesizing that for each subject $i$, we have
\begin{equation}
    \label{7}
    W_{itj} = m_{ij}(X_{itj}) + Z_{it}^{\prime}b_i + \varepsilon_{itj},\quad t=1,\ldots,n_i,
\end{equation}
and then cluster functions $\{m_{ij}(\cdot),i=1,\ldots,n\}$ into $m_j$ groups.

To fit the nonparametric univariate function in \eqref{7}, let $x \in [a,b]$ and let $(\xi_0=) a < \xi_1 < \xi_2 < \ldots < \xi_{J_n} <b (= \xi_{J_n+1})$ be a subdivision by $J_n$ distinct points on $[a,b]$, and the spline function $s(x)$ is a polynomial of degree $d$ on any interval $[\xi_{i-1}, \xi_{i}]$ which has $d-1$ continuous derivations on $(a,b)$. On the basis of B-spline, we can write a spline as $s(x,\gamma) = \sum_{l=1}^{J_n+d+1} \gamma_l B_l(x)$, where $\gamma = (\gamma_1,\ldots,\gamma_{J_n+d+1})^{\prime}$ is the vector of spline coefficients. Noting that $\bm W_{ij}=(W_{i1j},\ldots,W_{i n_i j})^{\prime}$, and $\bm B_{ij}= \{ B_l(x_{itj})\}^{l=1,\ldots,J_n+d+1}_{t=1,\ldots,n_i}$, also supposing that $\bm B^{\prime}_{ij} \bm B_{ij}$ is non-singular, the spline coefficients are estimated by
\begin{equation}
    \label{8}
    \widehat \gamma_{ij} = \mathop {argmin} \limits_{\gamma_{ij}} \frac {1}{n_i} \sum_{t=1}^{n_i} (W_{itj} - s(x_{itj},\gamma_{ij}))^2=[\bm B^{\prime}_{ij} \bm V_i \bm B_{ij}]^{-1} \bm B^{\prime}_{ij} \bm V_i \bm W_{ij}.
\end{equation}
Thanks to the fine properties of B-spline that the $n$ functions we estimate share the same degree and knots, as well the same basis functions $(B_{1},\ldots,B_{J_n+d+1})$, each coordinate $\gamma_{ij}$ has the same meaning for each function $m_{ij},i=1,\ldots,n$. Thus the set of functions $\{ m_{1j}(\cdot),\ldots,m_{nj}(\cdot) \}$ is summarized by $\{ \gamma_{1j},\ldots, \gamma_{nj} \}$, a set of vectors of $\mathbb R^{J_n+d+1}$, and we just need to partition their coefficients $\gamma_{ij}$ \citep{abraham2003unsupervised}.

For the set $\{ \gamma_{1j},\ldots, \gamma_{nj}\}$, we can easily establish the $n \times n$ distance matrix, most of time based on Euclidean distance. Numerous methods have been proposed to deal with this kind of clustering problem based on dissimilarity measure, such as k-means \citep{macqueen1967some}, fuzzy clustering \citep{bellman1966abstraction,ruspini1969new}, hierarchical clustering \citep{johnson1967hierarchical}, model-based clustering \citep{fraley2002model}. In the literature, $k$-means is one of the most popular clustering methods to make objects within clusters mostly homogeneous and objects between clusters mostly heterogeneous by minimizing the objective function
\begin{equation}
    \label{9}
    \sum_{s=1}^{k} \sum_{i \in \Omega_{s}}\left\|z_{i}-\mu_{s}\right\|_2^{2},
\end{equation}
where $\mu_s$ is the center of cluster $\Omega_{s}$ and $z_i$ is the $i$th object. %In the next section, the asymptotic properties of our proposed method is also established based on adopting $k$-means as clustering technique in this step.

After partitioning $\{i:1 \leq i \leq n \}$ into $\hat m_j$ subgroups on $X_j$, we conduct generalized least squares to the subjects who belong to the group $\widehat \Omega_{k,j},k=1,\ldots,\hat m_j$ to obtain $\widehat \gamma_{kj}^{(r)}$ in the $r$th iteration by fitting
\begin{equation}
    \label{10}
    \bm W_{ij} \approx \bm B_{ij} \gamma_{kj} + \bm Z_{i} b_i + \bm \varepsilon_i, \quad i \in \widehat \Omega_{k,j}.
\end{equation}
We repeat this backfitting procedure for $j=1,\ldots,p$ until the estimated subgroups do not change anymore, and thus the final identified subgroup structure yields. The final estimator of the additive components are denoted by $\widehat f_1,\ldots,\widehat f_p$, and the final estimated B-spline coefficients are denoted by $\widehat \gamma_1,\ldots,\widehat \gamma_p$.

\subsubsection{Determining the number of clusters}

Another concerning issue is determining the number of clusters $m_j$, as $k$-means procedure requires to pre-specify the number of clusters $m_j$. In general, model-selection criteria such as Akaike Information Criterion or Bayesian Information Criterion \citep{schwarz1978estimating} can be used to decide $m_j$ to avoid over-fitting.

AIC and BIC are both methods of assessing model fit penalized for the number of estimated parameters, but they could give different results for estimating the number of clusters in a dataset.  In cluster analysis, BIC tends to be preferred to AIC in estimating the number of clusters because it uses a larger penalty term and hence can recommend fewer clusters and it’s mathematical formulation is more
meaningful in this context \citep{Pelleg2000XmeansEK}, whereas AIC is more general.
BIC can be written as
\begin{equation}
    \label{11}
    BIC(k) =-2 log(\hat L) + log(n) k,
\end{equation}
where $log(\hat L)$ is the log-likelihood of the model, $k$ is the total number of parameters and $n$ is the sample size. An attractive property of the BIC is its consistency: the BIC selects the correct model with a probability that goes to $1$ as $n$ grows large \citep{zhang2010regularization,bai2018consistency}. Let $\hat{m}_{j}$ be the minimizer of BIC for choosing $m_j$, we show that  $\hat{m}_{j} \stackrel{P}{\rightarrow} m_{j}$ in the next section. In summary, the main steps of our subgroup identification algorithm is summarized in Algorithm 1.

\renewcommand{\algorithmicrequire}{\textbf{Input:}}
\renewcommand{\algorithmicensure}{\textbf{Output:}}
\begin{algorithm}[htb]
	\caption{Subgroup identification based on backfitting}
	\label{alg::conjugateGradient}
	\begin{algorithmic}[1]
		\REQUIRE ~~\\
		$Y_{it}$: the $t$th observation of the $i$th individual of response variable; \\
		$X_{it}$, the $t$th observation of the $i$th individual of time dependent variables;\\
		$Z_{it}$, the $t$th observation of the $i$th individual of covariates associated with random effects;\\
		$S_{i}$, baseline covariates of the $i$th individual;
		\ENSURE ~~\\
		group index $\widehat \Omega_{1j},\ldots,\widehat \Omega_{m_j j}, j=1,\ldots,p$;\\
		estimated coefficient $\widehat \beta$ and estimated function $\widehat f_1,\ldots, \widehat f_p$;
		\STATE \textbf{Initialize}: fit $Y_{it} = S_{i}^{\prime} \beta + f_1(X_{it1}) +\ldots + f_p(X_{itp}) + Z_{it}^{\prime} b_i+\varepsilon_{it}, i=1,\ldots,n,t=1,\ldots,n_i$;\\
		\STATE \textbf{return} initial estimate $\widehat f_1^{(0)},\ldots,\widehat f_p^{(0)},\widehat \beta$;\\
		\FOR{$j=1,\ldots,p$}
		\STATE set $W_{itj} = Y_{it} - S_{i} \widehat \beta - \sum_{k \neq j} \widehat f_k(X_{itk}) $;\\
		\FOR{$i=1,\ldots,n$}
		\STATE solve univariate nonparametric regression $W_{itj} = m_{ij}(X_{itj})+  Z_{it}^{\prime} b_i + \varepsilon_{it}$ with B-spline;\\
		\STATE \textbf{return} B-spline coefficients $\widehat \gamma_{ij}$;\\
		\ENDFOR
		\STATE clustering $\{\widehat \gamma_{1j},\ldots,\widehat \gamma_{nj}\}$ with k-means;\\
		\STATE partition $\{i \in 1,\ldots,n \}$ to group $\widehat \Omega_{1,j},\ldots,\widehat \Omega_{\widehat m_j,j }$;\\
		\STATE update $\widehat f_j = \widehat f_{1j}I(i \in \widehat \Omega_{1,j}) + \widehat f_{2j} I(i \in \widehat \Omega_{2,j}) + \ldots + \widehat f_{\widehat m_j j} I(i \in \widehat \Omega_{\widehat m_j,j})$;\\
		\ENDFOR
		\STATE repeat iteration line (3-12) until subgroups $\{\widehat \Omega_{1j},\ldots,\widehat \Omega_{\widehat m_j j}, j=1,\ldots,p\}$ converge;\\
		\RETURN subgroup index $\{\widehat \Omega_{1j},\ldots,\widehat \Omega_{\widehat m_j j}, j=1,\ldots,p\}$ and estimated $\widehat f_1,\ldots, \widehat f_p,\widehat \beta$.\\
   \end{algorithmic}
\end{algorithm}

\section{Asymptotic properties}

In this section, we establish the asymptotic properties of the proposed estimator. Let $L_2(\mathcal X)$ be the space of all square integrable functions on $\mathcal X = [0,1]$, and $\|f\|^2_2=\int_{0}^{1} f(x)^2 dx$ for any $f \in L_2(\mathcal X)$. Denote $\|f\|^2 = E[f(X)^2]$ and $\|f\|^2_n = \frac{1}{n} \sum_{i=1}^{n}f(X_i)^2$ as the theoretical and empirical norms respectively, where $X_i$ is a random sample from $\mathcal X$. For any set $\mathcal G$, $|\mathcal G|$ represents the cardinal of $\mathcal G$. For unbalanced dataset, we define $n_0 = \text{min}_i \{ n_i \}$, and $J_n$ is the number of internal knots. And several regularity conditions are required to establish the asymptotic properties.\\
(A1) The function $f_i(\cdot) \in C^r[0,1](i=1,\ldots,n)$ for some $r \geq 1$. \\
(A2) Let $a = \xi_0 < \xi_1 < \cdots < \xi_{J_n} < \xi_{J_{n+1}} =b$ be a partition of $[a,b]$ into $J_n+1$ subintervals, the knot sequences $\{ \xi_l \}_{l=0}^{J_n+1}$ have bounded mesh ratio, which means for some constant $C$
$$\frac{\max _{0 \leq l \leq J_{n}}\left|\xi_{l+1}-\xi_{l}\right|}{\min _{0 \leq l \leq J_{n}}\left|\xi_{l+1}-\xi_{l}\right|} \leq C.$$
(A3) The design points $\{x_{it}\} (i=1,\ldots,n;t=1,\ldots,T_i)$ follow an absolutely continuous density function $g_X$, and there exist constants $a_1$ and $a_2$ such that $0 < a_1 \leq \text{min}_{x \in \mathcal X} g_X(x) \leq \text{max}_{x \in \mathcal X} g_X(x) \leq a_2 < \infty$. \\
(A4) Assume that $N_g = O(N)$, where $N_{gj} = \sum_{i \in \Omega_{g,j}} n_i$ for $g=1,\ldots, m_j$, and $N_{0j} = min(N_{1j},\ldots, N_{m_j j})$, $N=\sum_{i=1}^{n} n_i$.\\
(A5) The eigenvalues of $\boldsymbol{\Sigma}=\operatorname{diag}\left(\boldsymbol{\Sigma}_{1}, \ldots, \boldsymbol{\Sigma}_{n}\right)$ and $\boldsymbol{V}=\operatorname{diag}\left(\boldsymbol{V}_{1}, \ldots, \boldsymbol{V}_{n}\right)$ are bounded away from zero to infinity.

Assumptions(A1)-(A3) are standard conditions for the nonparametric B-spline smoothing functions, and (A4) indicates that we require the cluster size to grow as the sample size increases.

\begin{theorem}
\label{th1}
Under Assumption (A1)-(A5), as $n \rightarrow \infty$, and given a sufficiently large $N_{0j}$ such that $J_n=O(N_{0j}^{\varsigma})$ with $0<\varsigma<1$, then the oracle estimator $\tilde \gamma_j$ satisfying
\begin{equation*}
    \| \tilde \gamma_j - \gamma_j\|_2^2 =O_p( J_n/N_{0j}+ J_n^{-2r}),j=1,\ldots,p,
\end{equation*}
where $\tilde \gamma_j$ denotes estimated spline coefficients when the true subgroup membership is known.
\end{theorem}

\begin{remark}
The result of Theorem \ref{th1} implies that the convergence rate of the oracle approximation $\tilde \gamma_j,j=1,\ldots,p$ is faster than the B-spline estimator $\widehat \gamma_j$. The convergence rate of oracle estimator assures that when prior knowledge on the true subgroup memberships is known, more information from each cluster with sufficient number of repeated measurements can be used. And the proof of Theorem \ref{th1} is similar to the proof of Lemma \ref{lem4}, which we will present in the Appendix.
\end{remark}

\begin{theorem}
\label{th2}
Under Assumption (A1)-(A5), as $n \rightarrow \infty$, and given a sufficiently large $n_0$ such that $J_n=O(n_{0}^{\varsigma})$ with $0<\varsigma<1$,, then the estimated additive components $\widehat f_j$ satisfying
\begin{equation*}
    \| \widehat f_j - f_j \|_2^2 = O_p(J_n/n_0 +J_n^{-2r}).
\end{equation*}
\end{theorem}

\begin{remark}
Theorem \ref{th2} shows that the convergence rate of the proposed additive estimator $\widehat f_j$ is of the same order as B-spline coefficients estimator $\widehat \gamma_{ij}^{(r)}$ in backfitting procedure. Theorem \ref{th2} holds given a sufficiently large number of repeated measurements, however in practice, it does not have to be such large, in simulations if $n_i$ is bigger than 7 or 8, we will have a nice result, more than 10 times repeated measurements would be better. The proof of Theorem \ref{th2} is given in the Appendix.
\end{remark}

\begin{theorem}
\label{th3}
Under the same conditions in Theorem \ref{th2}, as $n_0 \rightarrow \infty$ and $\| \widehat \gamma_{ij} - \gamma_{ij}\|_2^2 \rightarrow 0$, we have
\begin{equation*}
P(\widehat \Omega_{j} = \Omega_{j}) \rightarrow 1,
\end{equation*}
where $\widehat \Omega_j = \{\widehat \Omega_{1,j}, \ldots, \widehat \Omega_{m_j,j} \}$ is the estimated subgrouping membership on covariate $X_j$, $\Omega_j = \{ \Omega_{1,j}, \ldots, \Omega_{m_j,j} \}$ is the true subgrouping membership, and $\widehat \gamma_{ij}$ is the estimate of $\gamma_{ij}$.
\end{theorem}

\begin{remark}
Theorem \ref{th3} indicates that when there is a sufficient number of repeated measurements for each subject, the proposed method can identify the true subgroup with probability tending to 1.
\end{remark}

\section{Numerical studies}
In this section, we conduct Monte Carlo simulations on several examples to investigate the performance of the proposed subgroup identification procedure, and also compare the proposed method with some existing methods. Balanced and unbalanced longitudinal datasets are both considered, the situations of no subgroup (only one group), two subgroups are generated in different settings. Three cases are set up with the true additive components to be the combination of linear functions, the combination of nonlinear functions, and also the combination of linear functions and nonlinear function. For each simulation setting, we conduct the experiment for sample size $n=30,50,100,200$ and repeat 100 times for all cases, as well the working correlation matrix is chosen to be first-order autoregressive and exchangeable structure respectively (AR(0.3), AR(0.5), EX(0.3), EX(0.5)).

As for the evaluation of subgroup identification, we consider employing the accuracy percentage ($\%$), Rand Index (RI) \citep{rand1971objective} and Normalized Mutual Information (NMI) \citep{vinh2010information} to assess the identification performance, as they can measure how close the estimated grouping structure approaches the true structure. These three values are ranging between 0 and 1, with larger value indicating a higher score of similarity between the two groups. The accuracy percentage($\%$) is defined as the proportion of subjects that are correctly identified. And the definition of Rand Index is
\begin{equation*}
    RI = \frac{TP + TN}{TP+FP+FN+TN},
\end{equation*}
where TP (true positive) is the number of pairs fo subjects that are from different subgroups and assigned to different clusters, TN (true negative) is the number of pairs that are in the same subgroup and assigned to the cluster, FP (false positive) means the number of pairs that in the same subgroup but assigned to different clusters, FN (false positive) denotes the numbers of pairs which are from different subgroups but assigned to the same cluster.
Assuming that $A = \{A_1,A_2,\ldots \}$ and $B=\{B_1,B_2,\ldots \}$ are two sets of disjoint group of $\{1,2,\ldots,n\}$, the NMI is defined as
\begin{equation*}
    NMI(A,B) = \frac{2 I(A,B)}{H(A) + H(B)},
\end{equation*}
where
\begin{equation*}
    I(A, B)=\sum_{i, j} \frac{\left|A_{i} \cap B_{j}\right|}{n} \log \left(\frac{n\left|A_{i} \cap B_{j}\right|}{\left|A_{i} \| B_{j}\right|}\right)
\end{equation*}
is the mutual information between the two groups, and the entropy of $A$ is defined as following
\begin{equation*}
    H(A) = \sum_{i} \frac{\left|A_{i}\right|}{n} \log \left(\frac{n}{\left|A_{i}\right|}\right).
\end{equation*}

Furthermore, to investigate the estimation accuracy, we calculate the mean squared error (MSE) of the proposed approach and report the ratio of oracle MSE to MSE of our method. A ratio closer to 1 indicates a better estimation performance. And the following formula is utilized to find corresponding MSE of the estimated $\hat y_{it}$:
\begin{equation*}
    MSE = \frac {1}{n} \sum_{i=1}^{n} \frac{1}{n_i}\sum_{t=1}^{n_i}(\hat y_{it} -y_{it})^2.
\end{equation*}

\textbf {\emph{Case \uppercase\expandafter{\romannumeral1}}}: In this setting, we consider the situation both $f_1$ and $f_2$ are linear functions,
\begin{equation*}
    y_{it} = f_1(x_{it1}) +f_2(x_{it2}) + Z_{it}^{\prime} b_i + \varepsilon_{it},i=1,\ldots,n,t=1,\ldots,n_i
\end{equation*}
where the true subgroup structure is
\begin{equation*}
    f_{1}(x)=\left\{\begin{array}{ll}
	3 x-1.5,&1 \leq i \leq n / 2 \\
	-5 x+2.5,& n / 2<i \leq n
\end{array},
f_{2}(x)=\left\{\begin{array}{ll}
	1.25 x-0.625, &1 \leq i \leq n / 2 \\
	-6 x+3, &n / 2<i \leq n
\end{array}\right.\right.
\end{equation*}
The covariates $X_{it1}, X_{it2}$ are both generated from $U(0,1)$, $Z_{it} \sim N(0,1)$. The vector of random effects $b_i \sim N(0,0.2)$ and $\varepsilon_{it}$ are generated from $N(0,0.1)$. We compare our method with two existing approaches, the model-based clustering and classification for longitudinal data method (\emph{longclust}) \citep{mcnicholas2016model}, which is based on gaussian model mixture to detect subgroups, and the k-means for longitudinal data method (\emph{kml}) \citep{genolini2010kml}, which implements k-means clustering based on generalized Fréchet distance to achieve subgrouping. Both \emph{longclust} approach and \emph{kml} approach are not suitable for unbalanced data, so in this setting we generate balanced data with $n_i=15, i=1,\ldots,n$.

Table 1 shows that the accuracy percentage ($\%$), NMI and Rand Index of the proposed method are equal to 1, which means the proposed method has a excellent performance in the linear setting considering the sample size $n=30, 50$ are relatively small. In addition, for \emph{kml} method, the accuracy percentage and Rand Index are around $50\%$, indicating it only correctly identified half of the sample. And the performance of longclust method is quite different under different sample sizes, when the sample size is larger than 100, it has a quite well performance, but for $n=30,50$, its performance become unsatisfactory, indicating its inapplicability to small sample size. As \emph{longclust} and \emph{kml} methods only use the information of response variable and time, and the proposed method includes other covariates and random effect to modeling the trajectories and pursuing subgroups, the proposed method can better characterize the structure of longitudinal data.

\begin{table}[] \centering
	\newcommand{\tabincell}[2]{\begin{tabular}{@{}#1@{}}#2\end{tabular}}
	\caption{Results for Case I. The accuracy percentage ($\%$), NMI, Rand Index for the proposed method, \emph{longclust} method and \emph{kml} method with working correlation matrix to be AR(0.3), AR(0.5), EX(0.3) and EX(0.5). The sample size is chosen to be 30, 50, 100 and 200 respectively. The values in the parentheses are the standard deviations.}
	\label{}
	\resizebox{\textwidth}{!}{
	\begin{tabular}{@{\extracolsep{5pt}}lccccccc}
		\\[-1.8ex]\hline
		\hline \\[-1.8ex]
		&\multirow{2}{*}{n}& \multicolumn{4}{c}{proposed method} &\multirow{2}{*}{longclust} &\multirow{2}{*}{kml}\\
		\cline{3-6} \\[-1.8ex]
		\multicolumn{1}{c}{} &  & \multicolumn{1}{c}{AR(0.3)}  &\multicolumn{1}{c}{AR(0.5)} & \multicolumn{1}{c}{EX(0.3)} & \multicolumn{1}{c}{EX(0.5)} \\
		\hline \\[-1.8ex]
		
		\multirow{8}{*}{$\%$}& 30 & 1 & 1 & 1 & 1 & 0.5496 & 0.5112  \\
		                     &    &(0)&(0)&(0)&(0)&(0.0040)  &(0.0120) \\
		                     & 50 & 1 & 1 & 1 & 1 & 0.600 & 0.5206  \\
		                     &    &(0)&(0)&(0)&(0)&(0.0030)  &(0.0060) \\
		                     & 100& 1 & 1 & 1 & 1 & 0.9330 & 0.5397  \\
		                     &    &(0)&(0)&(0)&(0)&(0.0024)&(0.0030)\\
		                     & 200& 1 & 1 & 1 & 1 & 0.9410 & 0.5667 \\
		                     &    &(0)&(0)&(0)&(0)&(0.0025)&(0) \\
		\hline \\[-1.8ex]
		\multirow{8}{*}{NMI}& 30 & 1 & 1 & 1 & 1 & 0.1016 &0.0050  \\
		                    &    &(0)&(0)&(0)&(0)&(0.0210)&(0.0127) \\
		                    & 50 & 1 & 1 & 1 & 1 & 0.1505 &0.0191\\
		                    &    &(0)&(0)&(0)&(0)&(0.012) &(0.0048)\\
		                    & 100& 1 & 1 & 1 & 1  &0.6936 & 0.0369\\
		                    &    &(0)&(0)&(0)&(0)&(0.0080)&(0.0017)\\
		                    & 200& 1 & 1 & 1 & 1  &0.7067 &0.0374  \\
		                    &    &(0)&(0)&(0)&(0)&(0.0030)&(0)\\		
		\hline \\[-1.8ex]
		\multirow{8}{*}{Rand Index}& 30 & 1 & 1 & 1 & 1 & 0.5000 &0.4852  \\
		                          &    &(0)&(0)&(0)&(0)&(0.0062)&(0.0183) \\
		                          & 50 & 1 & 1 & 1 & 1 & 0.5014 &0.4912\\
		                          &    &(0)&(0)&(0)&(0)&(0.0041) &(0.0117)\\
		                          & 100& 1 & 1 & 1 & 1  &0.8770 & 0.4953\\
		                          &    &(0)&(0)&(0)&(0)&(0.0024)&(0.0008)\\
		                          & 200& 1 & 1 & 1 & 1  &0.8916 &0.5026  \\
		                          &    &(0)&(0)&(0)&(0)&(0.0020)&(0.0004)\\		
		
		\hline \\[-1.8ex] 					
	\end{tabular}
}
\end{table}

\textbf {\emph{Case \uppercase\expandafter{\romannumeral2}}}: In this setting, we consider the scenario of unbalanced data, and let both $f_1$ and $f_2$ to be nonlinear functions,
\begin{equation*}
    y_{it} = f_1(x_{it1}) +f_2(x_{it2}) + Z_{it}^{\prime} b_i + \varepsilon_{it},i=1,\ldots,n,t=1,\ldots,n_i,
\end{equation*}
where the true subgroup structure is
\begin{equation*}
    f_1(x)=\begin{array}{l}
\begin{cases}
-1.75  arctan(5(x-0.6))-0.415,  &\text{when } i \text{ is odd}\\
2.5  (1-((x-0.75)/0.8)^2)^4-1.363, &\text{when } i \text{ is even}
\end{cases}
\end{array}
\end{equation*}
\begin{equation*}
    f_2= \begin{array}{l}
\begin{cases}
 2  sin(\pi x) -\frac{\pi}{4},  &1\leq i\leq n\\
 2  cos(\pi x), &n/2<i\leq n
\end{cases}
\end{array}
\end{equation*}
The covariates $X_{it1}, X_{it2}$ are both generated from $U(0,1)$, $Z_{it} \sim N(0,1)$. The vector of random effects $b_i \sim N(0,0.2)$ and $\varepsilon_{it}$ are generated from $N(0,0.1)$. $n_i$ is a random integer generated between 10 and 20. Since the \emph{kml} method and \emph{longclust} method are not suitable for this scenario, we only present the simulation result of the proposed approach.

Table 2 shows the simulation results for Case II. For each setup, we present the accuracy percentage on each covariate ($\%_{X_1},\%_{X_2}$) and total accuracy percentage ($\%_{total}$), as well as the NMI and Rand Index. Generally, when sample size $n$ increases, the performance of subgroup identification and estimation also shows an increasing trend. Each performance evaluation index is close to 1, which reflects the good clustering accuracy of the proposed method. %It makes sense that the simulation results for Case II are not as good as Case I, since the setting of nonlinear functions is more complicated than the linear one, and leads to the difficulty of subgroup identification increasing. 
Moreover, the MSE values of the proposed method are comparable to the oracle values, indicating our approach is able to detect subgroups and give a good estimate simultaneously.

\begin{table}[!htbp] \centering
	\newcommand{\tabincell}[2]{\begin{tabular}{@{}#1@{}}#2\end{tabular}}
	\caption{Results for Case II. The accuracy percentage on $X_1$ ($\%_{X_1}$), the accuracy percentage on $X_2$ ($\%_{X_2}$), the accuracy percentage of total population ($\%_{\text{total}}$), NMI, Rand Index and the ratio of oracle MSE to MSE of the proposed method with working correlation matrix to be AR(0.3), AR(0.5), EX(0.3) and EX(0.5). The sample size is chosen to be 30, 50, 100 and 200 respectively. The values in the parentheses are the standard deviations.}
	\label{}
	\begin{tabular}{@{\extracolsep{5pt}}lccccccc}
		\\[-1.8ex]\hline
		\hline \\[-1.8ex]
		\multicolumn{1}{c}{} & \multicolumn{1}{c}{n} & \multicolumn{1}{c}{$\%_{X_1}$} & \multicolumn{1}{c}{$\%_{X_2}$}
		&\multicolumn{1}{c}{$\%_{total}$} &\multicolumn{1}{c}{NMI} &\multicolumn{1}{c}{Rand Index} &\multicolumn{1}{c}{Ratio}  \\
		\hline \\[-1.8ex]
		
		\multirow{4}{*}{\textbf{AR(0.3)}} & 30 & 0.9863 & 0.9924  & 0.9850 & 0.9665 & 0.9859 & 0.9551  \\
		&    &(0.0268)&(0.0144) &(0.0269)&(0.0513)&(0.0240)&(0.2054) \\
		& 50 & 0.9882 & 0.9932  & 0.9870 & 0.9692 & 0.9868 & 0.9667 \\
		&    &(0.0148)&(0.0116) &(0.0169)&(0.0419)&(0.0169)&(0.2002)\\
		& 100& 0.9906 &0.9945   &0.9905  &0.9696  &0.9904  &0.9744 \\
		&    &(0.0093)&(0.0069) &(0.0098)&(0.0295)&(0.0098)&(0.1544)\\
		& 200& 0.9913 &0.9953   &0.9910  &0.9711  &0.9910  &0.9843 \\
		&    &(0.0072)&(0.0067)&(0.0075) &(0.0240)&(0.0074)&(0.0808)\\
		
		\cline{2-8}
		\multirow{4}{*}{\textbf{AR(0.5)}} & 30 & 0.9850 & 0.9920  & 0.9833 & 0.9639 & 0.9837 & 0.9524  \\
		&    &(0.0203)&(0.0157) &(0.0209)&(0.0447)&(0.0204)&(0.1986) \\
		& 50 & 0.9906 & 0.9936  & 0.9884 & 0.9642 & 0.9886 & 0.9612 \\
		&    &(0.0134)&(0.0152) &(0.0173)&(0.0381)&(0.0166)&(0.1343)\\
		& 100& 0.9912 &0.9942   &0.9892  &0.9675  &0.9888  &0.9631 \\
		&    &(0.0094)&(0.0099) &(0.0114)&(0.0309)&(0.0118)&(0.1199)\\
		& 200& 0.9914 &0.9946   &0.9902  &0.9724  &0.9895  &0.9701 \\
		&    &(0.0063)&(0.0055)&(0.00062)&(0.0219)&(0.0066)&(0.0645)\\
		
		\cline{2-8}
		\multirow{4}{*}{\textbf{EX(0.3)}} & 30 & 0.9970 & 0.9980  & 0.9970 & 0.9930 & 0.9970 & 0.9985  \\
		&    &(0.0096)&(0.0079) &(0.0096)&(0.0211)&(0.0096)&(0.0287) \\
		& 50 & 0.9977 & 0.9987  & 0.9977 & 0.9934 & 0.9977 & 0.9920 \\
		&    &(0.0046)&(0.0034) &(0.0047)&(0.0142)&(0.0046)&(0.0263)\\
		& 100& 0.9983 & 0.9992  &0.9983  &0.9940  &0.9983  &0.9927 \\
		&    &(0.0042)&(0.0030) &(0.0042)&(0.0112)&(0.0041)&(0.0231)\\
		& 200& 0.9991 &0.9995   &0.9991  &0.9976  &0.9991  &0.9953 \\
		&    &(0.0034)&(0.0021) &(0.0034)&(0.0112)&(0.0034)&(0.0195)\\
		
		\cline{2-8}
		\multirow{4}{*}{\textbf{EX(0.5)}} & 30 & 0.9970 & 0.9983  & 0.9970 & 0.9918 & 0.9970 & 0.9880  \\
		&    &(0.0096)&(0.0073) &(0.0096)&(0.0210)&(0.0095)&(0.0417) \\
		& 50 & 0.9975 & 0.9984  & 0.9975 & 0.9934 & 0.9974 & 0.9888 \\
		&    &(0.0067)&(0.0054) &(0.0067)&(0.0172)&(0.0067)&(0.0379)\\
		& 100& 0.9976 & 0.9988  &0.9976  &0.9935  &0.9977  &0.9891 \\
		&    &(0.0037)&(0.0025) &(0.0037)&(0.0118)&(0.0037)&(0.0229)\\
		& 200& 0.9986 &0.9993   &0.9986  &0.9958  &0.9986  &0.9931 \\
		&    &(0.0034)&(0.0023) &(0.0034)&(0.0109)&(0.0034)&(0.0203)\\
		\hline \\[-1.8ex] 					
	\end{tabular}
\end{table}

\textbf {\emph{Case \uppercase\expandafter{\romannumeral3}}}: In this setting, we consider the scenario that $f_1$ and $f_3$ are nonlinear functions, and $f_2$ are linear function, and we include variable $U$ as baseline covariate.
\begin{equation*}
   y_{it} = u_i \beta + f_1(x_{it1}) +f_2(x_{it2}) +f_3(x_{it3})+ b_i + \varepsilon_{it}, i=1,\ldots,n,t=1,\ldots,n_i,
\end{equation*}
where the true subgroup structure is $$f_1(x)= 2 cos(\pi x),\quad i=1,\ldots,n$$
$$f_2(x)=\begin{array}{l}
\begin{cases}
3x-1.5, &\text{when } i \text{ is odd}\\
-5x+2.5,  &\text{when } i \text{ is even}
\end{cases}
\end{array}$$
$$f_{3}(x)=\begin{array}{l}
\begin{cases}
-1.75 \times arctan(5(x-0.6))-0.415 , & 1\leq i\leq n/2 \\
2.5 \times (1-((x-0.75)/0.8)^2)^4-1.363,  & n/2<i\leq n
\end{cases}
\end{array}.$$
The covariates $X_{it1}, X_{it2},X_{it3}$ are generated from $U(0,1)$, $U_i \sim B(1,0.5)$, well the vector of random effects $b_i \sim N(0,0.2)$ and $\varepsilon_{it}$ are generated from $N(0,0.1)$. $n_i$ is a random integer generated between 10 and 20.

Table 3 presents the accuracy percentage ($\%$), NMI and Rand Index of our method when the true model is set with a mixture of linear functions and nonlinear functions, and when there may not exist subgroup. All these indexes are close to 1, demonstrating the good performance of the proposed method. Furthermore, we can see from Table 3 that the MSE values of the proposed method are close to the MSE of oracle estimators. The simulation of this scenario suggests that whether the relationship between the covariates and response variable is nonlinear or linear function, or whether there exist subgroups or not, the proposed approach could successfully identify the subgroup membership and then give consistent estimate of each additive component.

As the preceding examples, the proposed method is not only suitable for balanced data, but also applicable of unbalanced scenario.  Furthermore, besides identifying subgroups with high accuracy, the proposed method also gives good estimation of additive components. In a word, these results confirm that the proposed approach has good performances on finding subgroups of heterogeneous trajectories.

\begin{table}[!htbp] \centering
	\newcommand{\tabincell}[2]{\begin{tabular}{@{}#1@{}}#2\end{tabular}}
	\caption{Results for Case III. The accuracy percentage on $X_1$ ($\%_{X_1}$), the accuracy percentage on $X_2$ ($\%_{X_2}$), the accuracy percentage on $X_3$ ($\%_{X_3}$), the accuracy percentage of total population ($\%_{\text{total}}$), NMI, Rand Index and the ratio of oracle MSE to MSE of the proposed method with working correlation matrix to be AR(0.3), AR(0.5), EX(0.3) and EX(0.5). The sample size is chosen to be 30, 50, 100 and 200 respectively. The values in the parentheses are the standard deviations.}
	\label{}
	\resizebox{\textwidth}{!}{
		\begin{tabular}{@{\extracolsep{5pt}}lcccccccccc}
			\\[-1.8ex]\hline
			\hline \\[-1.8ex]
			\multicolumn{1}{c}{} & \multicolumn{1}{c}{n} & \multicolumn{1}{c}{$\%_{X_1}$} & \multicolumn{1}{c}{$\%_{X_2}$}
			& \multicolumn{1}{c}{$\%_{X_3}$} &\multicolumn{1}{c}{$\%_{total}$} &\multicolumn{1}{c}{NMI} &\multicolumn{1}{c}{Rand Index} &\multicolumn{1}{c}{Ratio}  \\
			\hline \\[-1.8ex]
			
			\multirow{4}{*}{\textbf{AR(0.3)}} & 30 &1  & 0.9907 & 0.9893  & 0.9880 & 0.9745 & 0.9885 & 0.9210  \\
			&    &(0)&(0.0158)&(0.0206) &(0.0209)&(0.0435)&(0.0198)&(0.1471) \\
			& 50 &1 & 0.9927 & 0.9927  & 0.9924 & 0.9745 & 0.9927 & 0.9285 \\
			&    &(0)&(0.0111)&(0.0145) &(0.0152)&(0.0339)&(0.0142)&(0.1417)\\
			& 100&1 & 0.9933 &0.9932   &0.9925  &0.9789  &0.9926  &0.9519 \\
			&    &(0)&(0.0089)&(0.0093) &(0.0094)&(0.0281)&(0.0093)&(0.1359)\\
			& 200&1 & 0.9942 &0.9934   &0.9930  &0.9817  &0.9930  &0.9765 \\
			&    &(0)&(0.0056)&(0.0056)&(0.0058) &(0.0189)&(0.0058)&(0.0906)\\
			
			\cline{2-9}
			\multirow{4}{*}{\textbf{AR(0.5)}} & 30 &1 & 0.9918 & 0.9920  & 0.9833 & 0.9639 & 0.9837 & 0.9524  \\
			&    &(0)&(0.0145)&(0.0157) &(0.0209)&(0.0447)&(0.0204)&(0.1986) \\
			& 50 &1 & 0.9940 & 0.9936  & 0.9884 & 0.9642 & 0.9886 & 0.9612 \\
			&    &(0)&(0.0121)&(0.0152) &(0.0173)&(0.0381)&(0.0166)&(0.1343)\\
			& 100&1 & 0.9953 &0.9942   &0.9892  &0.9675  &0.9888  &0.9631 \\
			&    &(0)&(0.0061)&(0.0099) &(0.0114)&(0.0309)&(0.0118)&(0.1199)\\
			& 200&1 & 0.9955 &0.9946   &0.9902  &0.9724  &0.9895  &0.9701 \\
			&    &(0)&(0.0063)&(0.0055)&(0.00062)&(0.0219)&(0.0066)&(0.0645)\\
			
			\cline{2-9}
			\multirow{4}{*}{\textbf{EX(0.3)}} & 30 &1 & 0.9986 & 0.9989  & 0.9986 & 0.9930 & 0.9951 & 0.9948  \\
			&    &(0)&(0.0051)&(0.0035) &(0.0051)&(0.0211)&(0.0129)&(0.0231) \\
			& 50 &1 & 0.9986 & 0.9992  & 0.9986 & 0.9934 & 0.9957 & 0.9949 \\
			&    &(0)&(0.0046)&(0.0035) &(0.0047)&(0.0142)&(0.0123)&(0.0196)\\
			& 100&1 & 0.9986 & 0.9993  & 0.9986 &0.9940  &0.9965  &0.9955 \\
			&    &(0)&(0.0041)&(0.0035) &(0.0041)&(0.0112)&(0.0101)&(0.0149)\\
			& 200&1 & 0.9993 & 0.9997   &0.9993 &0.9976  &0.9986  &0.9977 \\
			&    &(0)&(0.0025)&(0.0021) &(0.0026)&(0.0112)&(0.0090)&(0.0110)\\
			
			\cline{2-9}
			\multirow{4}{*}{\textbf{EX(0.5)}} & 30 &1 & 0.9983 & 0.9987  & 0.9964 & 0.9953 & 0.9969 & 0.9911  \\
			&    &(0)&(0.0072)&(0.0065) &(0.074) &(0.0142)&(0.0088)&(0.0371) \\
			& 50 &1 & 0.9984 & 0.9991  & 0.9983 & 0.9957 & 0.9983 & 0.9939 \\
			&    &(0)&(0.0056)&(0.0033) &(0.0056)&(0.0127)&(0.0055)&(0.0305)\\
			& 100&1 & 0.9985 & 0.9994  &0.9985  &0.9961  &0.9985  &0.9948 \\
			&    &(0)&(0.0042)&(0.0031) &(0.0042)&(0.0096)&(0.0042)&(0.0193)\\
			& 200&1 & 0.9989 &0.9994   &0.9989  &0.9984  &0.99989  &0.9962 \\
			&    &(0)&(0.0027)&(0.0018) &(0.0027)&(0.0072)&(0.0027)&(0.0103)\\
			\hline \\[-1.8ex] 		
			
		\end{tabular}
	}
\end{table}

\section{Application}

In this section, we apply the proposed method to primary biliary cirrhosis (PBC) data  collected between 1974 and 1984, which is available in R package ``joineRML". PBC is a chronic disease characterized by inflammatory destruction of the small bile ducts within the liver, which eventually leads to cirrhosis of the liver, followed by death. Patients often present abnormalities in their blood tests, such as elevated and gradually increased serum bilirubin. Characterizing the patterns of time courses of bilirubin levels is medical interest.
A total of 424 patients, referred to Mayo Clinic during that ten-year interval, met eligibility criteria for the randomized placebo-controlled trial of the drug D-penicilamine, 312 formal study participants, and 106 eligible nonenrolled subjects. This dataset contains multiple laboratory results, but only on the first 312 patients, among which 140 had died and the rest were censored and the sex ratio is at least $9:1$ (women to men).
Some baseline covariates such as age, gender, drug use indicator, were recorded at the beginning of the study. And various biomarkers were collected longitudinally, containing presence of hepatomegaly, spiders, edema and ascites, serum bilirubin, albumin, alkaline, phosphatase, serum glutamic-oxaloacetic transaminase (SGOT), platelets per cubic, prothrombin time and histologic stage of disease, etc. We would not study the problem of missing data in this paper, and the covariate cholesterol includes too much missing data, so we will ignore this covariate. Considering that we will employ splines in the subgrouping procedure, subjects whose repeated measurements less than or equal to 8 times will not be included, thus in our analysis 62 patients with total 687 records will be utilized.

Biomedical research indicates that serum bilirubin concentration is a primary indicator to help assess and track the absence of liver diseases. It is usually normal at diagnosis (0.1-1 mg/dl) but rises with histological disease progression \citep{talwalkar2003primary}. Therefore, we focus on modeling the relationship between serum bilirubin and other covariates of interest. And we set log-transformed serum bilirubin level as the response variable, since the original level has positive observed values \citep{murtaugh1994primary}.

There are some issues we intend to investigate (i) how the bilirubin levels evolve over time; (ii) how the evolution of bilirubin levels is related to other biomarkers; (iii) whether there exist any subgroups with different additive components among subjects and if there exist subgroups, how they differ from each other. Based on previous researches \citep{su2008interaction,ding2008modeling,tang2019discrete}, time-independent variables, \emph {age}(transformed to dummy variable) at registration, \emph{gender} (male and female), \emph{drug type} (placebo and D-pencillamine) are included in our model, as well the time-dependent variables, \emph{time, SGOT} and \emph {prothrombin}. We analyse the PBC dataset based on the following semiparametric additive mixed effect model:
\begin{equation*}
\begin{split}
log(bili_{it}) = & \alpha + \beta_1 gender_i + \beta_2 drug_i + \beta_3 age_i \notag\\
& + f_1(time_{it}) +f_2(log(SGOT_{it})) +f_3(log(prothrombin_{it}))+b_i,
\end{split}
\end{equation*}
where $i=1,\ldots,n, t=1,\ldots,n_i$. $b_i$ is random effect, and covariate \emph{time} is rescaled to $[0,1]$, \emph{gender} is binary gender indicator with 1 for female, \emph{drug} is binary treatment indicator with 1 for D-pencillamine.
After applying the proposed subgroup identification approach to PBC dataset, 3 subgroups have been found on covariate \emph{time}, and the BIC value of 3 subgroups reported by our algorithm is obviously less than other situations. Meanwhile, no subgroups have been detected on variables \emph{SGOT} and \emph{prothrombin}. There are 34 individuals in $\widehat \Omega_{1,1}$, 12 individual in $\widehat \Omega_{2,1}$ and 16 individuals in $\widehat \Omega_{3,1}$. Figure 1 presents the diverse functions of $\widehat f_1(time_{it})$, and it is obvious that these three functions have completely different trends. The time effect for people in subgroup 1 is always negative for $t \in [0,1]$, whereas the time effect for patients in subgroup 2 developed rapidly, indicating the severe condition of their illness, meanwhile the time effect for individuals in subgroup 3 developed towards a worse situation too, but not such severe like patients in subgroup 2.
\begin{figure}
	\centering
	\includegraphics[scale=0.42]{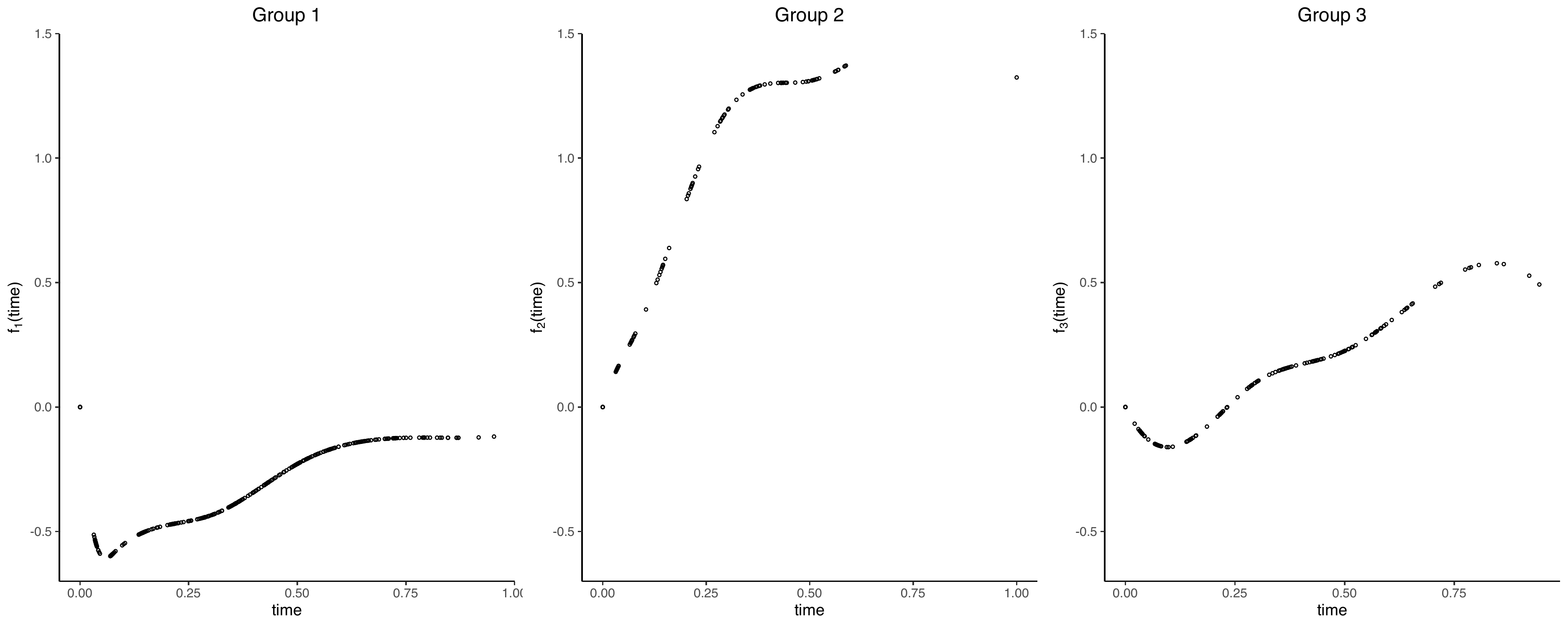}
	\caption{Plot of fitted $f(time)$ versus \emph{time} for subjects in Group 1, Group 2 and Group 3.}
	\label{fig}
\end{figure}

\begin{figure}
\centering
\includegraphics[scale=0.55]{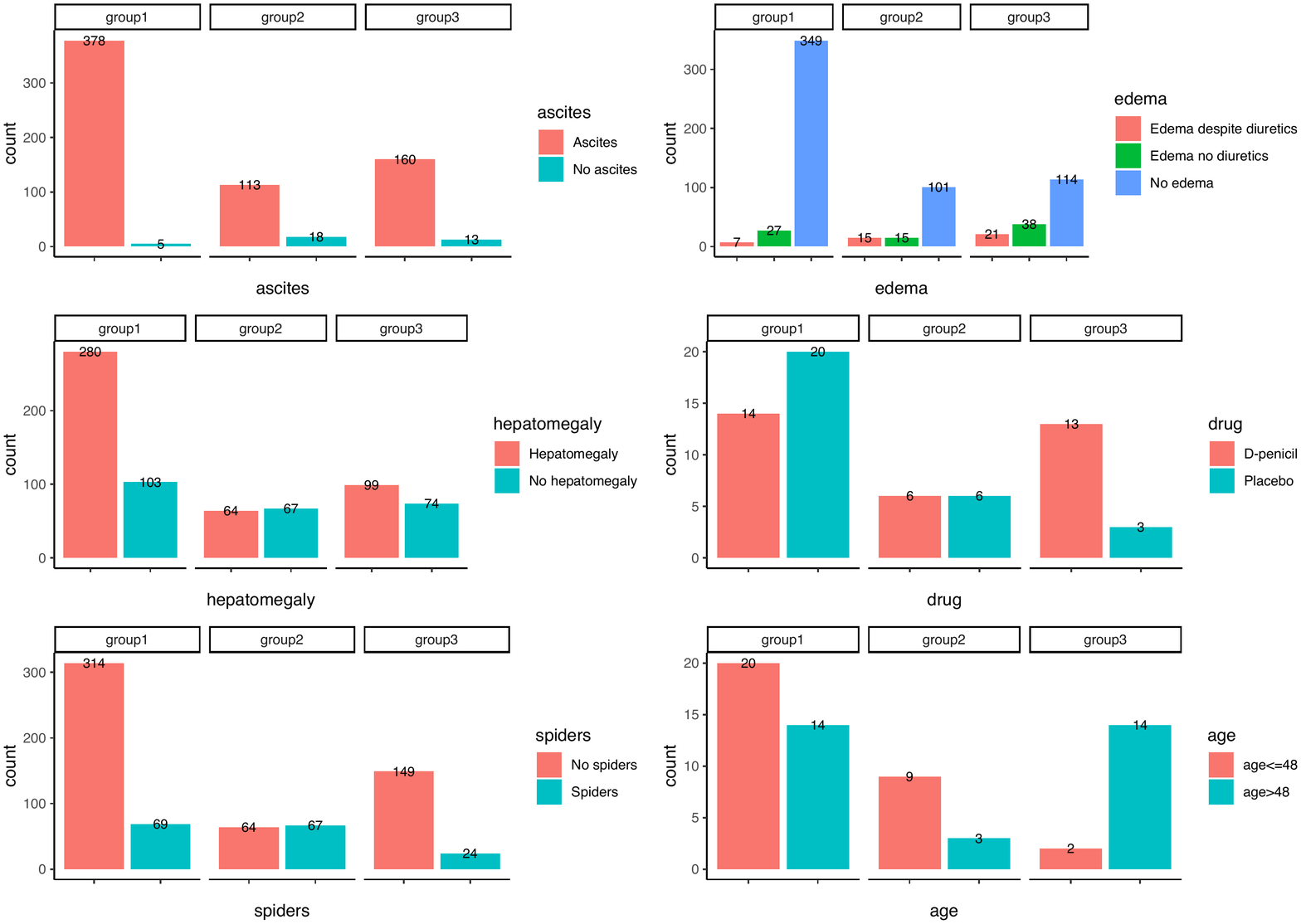}
\caption{The frequency of some baseline covariates and biomarkers (ascities, hepatomegaly, spiders, edema, drug, age) for subjects in different subgroups.}
\label{fig}
\end{figure}
The differences of the important biomarkers between three subgroups are shown in Figure 2 and Figure 3. We summarize the discrete biomarkers in histograms and the continuous biomarkers in box plots. We also conduct a $t$-test to compare the differences, and the result is displayed in Table 4. We could see that for some binary symptoms, such as \emph{ascites, hepatomegaly, spiders} and \emph{edema}, there are significant differences between these subgroups, demonstrating the heterogeneity of the subgroups detected by the proposed procedure. In addition, for the laboratory indexes collected longitudinally, like \emph{logserBilir, albumin, alkaline, SGOT, prothrombin}, there are also distinct differences between each subgroup. In particular, the average level of response variable \emph{logserBilir} for subgroup 2 is obviously higher than the patients in other groups, revealing their severity of illness, and this is consistent with our finding of $\widehat f_1(time)$ in Figure 1. The estimated model for each group is as follows:
\begin{equation*}
    \begin{aligned}
		\widehat {logbili}_{it} = & 0.833 -0.662 gender_i + 0.18 drug_i -0.177 age_i + \widehat f_{11}(time_{it}) I(i \in \widehat\Omega_{1,1}) \\
		& + \widehat f_{21}(time_{it}) I(i \in \widehat\Omega_{2,1})+ \widehat f_{31}(time_{it}) I(i \in \widehat\Omega_{3,1})\\
		& +\widehat f_2(log(SGOT_{it}))  + \widehat f_3(log(prothrombin_{it})) + \widehat b_i,\\
\end{aligned}
\end{equation*}
where
$$\begin{aligned}
& f_{11} = -0.692 B_{0,5} -0.446 B_{1,5} -0.505 B_{2,5} -0.055 B_{3,5} -0.138 B_{4,5} -0.119 B_{5,5}, \\
& f_{21} = 0.141 B_{0,5} +0.365 B_{1,5} + 1.402 B_{2,5} + 1.152 B_{3,5} + 1.861 B_{4,5} + 1.324 B_{5,5},\\
& f_{31} = -0.097 B_{0,5} -0.265 B_{1,5} + 0.188 B_{2,5} + 0.147 B_{3,5} + 0.755 B_{4,5} + 0.492 B_{5,5},\\
& f_2=-1.759 B_{0,5} - 1.057 B_{1,5} -0.692 B_{2,5} -0.171 B_{3,5} + 0.682 B_{4,5} + 0.082 B_{5,5},\\
& f_3= -0.072B_{0,5} -0.034 B_{1,5} -0.111 B_{2,5} +0.318 B_{3,5} + 1.150 B_{4,5} +1.275 B_{5,5}.\\
\end{aligned}$$

Throughout our method, the overall subjects have been partitioned to three segmentations. In this procedure, we take multiple covariates into consideration, and reveal the different time effect across three subgroups. In addition, we give the estimates of regression functions of each additive component in the model, which could capture the functional relationship between covariates and the response variable. The results provided by our approach may lead to more accurate subgrouping rule and may be helpful in making personalized medical decision for the patients who suffer from the disease.
\begin{figure}[h]
	\centering
	\includegraphics[scale=0.55]{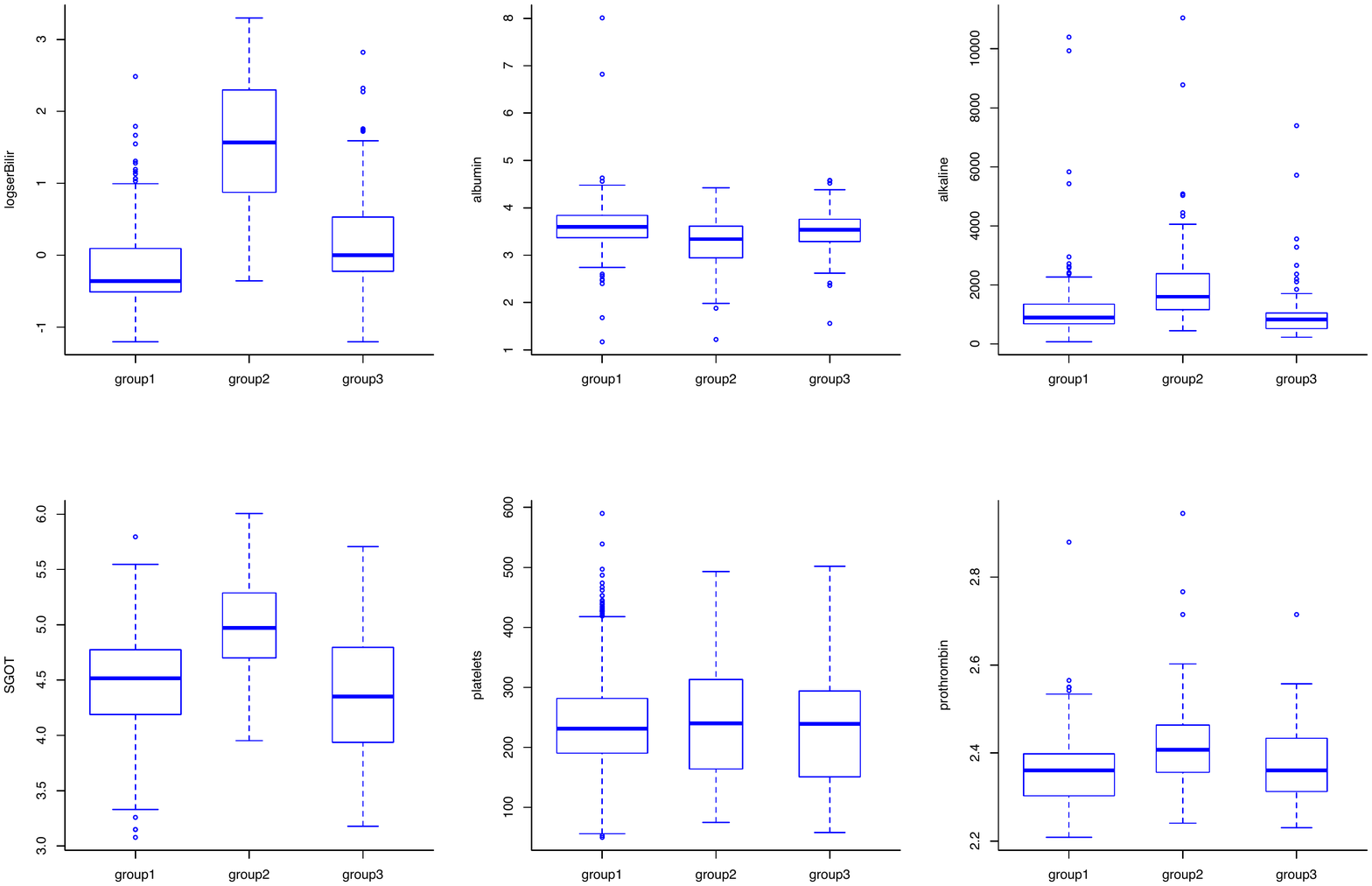}
	\caption{Box plots of biomarkers (\emph{logserBilir, albumin, alkaline, SGOT, platelets, prothrombin}) for subjects in different subgroups.}
\end{figure}

\begin{table}[!htbp] \centering
	\caption{The result of multiple t-test of the important biomarkers for individuals in different groups. $\mu_1,\mu_2$ and $\mu_3$ represent the mean of biomarkers of each subgroup. $p\text{-value}_{ij}$ is the $p$-value of multiple hypothesis test between group $i$ and group $j$.}
	\label{}
	\resizebox{\textwidth}{!}{
	\begin{tabular}{@{\extracolsep{5pt}}lcccccc}
		\\[-1.8ex]\hline
		\hline \\[-1.8ex]
		\multicolumn{1}{c}{Biomarkers} & \multicolumn{1}{c}{$\mu_1$} & \multicolumn{1}{c}{$\mu_2$} & \multicolumn{1}{c}{$\mu_3$}
		& \multicolumn{1}{c}{$\text{p-value}_{12}$} &\multicolumn{1}{c}{$\text{p-value}_{13}$} &\multicolumn{1}{c}{$\text{p-value}_{23}$}\\
		\hline \\[-1.8ex]
        ascites & 0.0131 & 0.1374 & 0.0751 & $\bm{7.6\times 10^{-8}}$ & $\bm{0.0039}$ & $\bm{0.0139}$ \\
        hepatomegaly & 0.2689 &	0.5115 & 0.4277 & $\bm{1.3 \times 10^{-6}}$ & $\bm{4.7 \times 10^{-4}}$ & 0.1237 \\
        spiders & 0.1802 & 0.5115 & 0.1387 & $\bm{4.7\times 10^{-15}}$ & 0.26 & $\bm{8.8\times 10^{-15}}$\\
        edema & 0.0535 & 0.1718 & 0.2312 & $\bm{2.8\times 10^{-5}}$ & $\bm{3.1 \times 10^{-12}}$ & 0.055 \\
        logserBilir & -0.2152 & 1.5681 & 0.2108 & $\bm{<2\times 10^{-16}}$ & $\bm{1.7 \times 10^{-12}}$ & $\bm{<2\times 10^{-16}}$ \\
        albumin & 3.6022 & 3.2886 & 3.5040 & $\bm{1.4\times 10^{-9}}$ & $\bm{0.0291}$ & $\bm{3.2 \times 10^{-4}}$ \\
        alkaline & 1106.57 & 1955.98 & 943.60 & $\bm{2.8 \times 10^{-16}}$ & 0.073 & $\bm{<2\times 10^{-16}}$\\
        SGOT & 4.4776 & 4.9641 & 4.3792 & $\bm{<2\times10^{-16}}$ & $\bm{0.03}$ & $\bm{<2\times 10^{-16}}$ \\
        prothrombin & 2.3594 & 2.4185 & 2.3722 & $\bm{2.1 \times 10^{-12}}$ & 0.087 & $\bm{2.3 \times 10^{-6}}$\\

		\hline \\[-1.8ex]
	\end{tabular}
	}
\end{table}

\section{Discussion}
This article introduced a novel framework of subgroup identification for longitudinal data, based on semiparametric additive mixed effect model. We aim at finding subgroups on each covariate to reflect the diverse relationship between each covariate with the response variable. It is of great interest to describe the various association between covariates and the response variable, which could reveal how each covariate attributes to the subgroups.  Numerical studies indicate that the proposed approach is effective in identifying subgroups and estimating the nonparametric regression functions simultaneously, for both balanced data and unbalanced data. %This new subgrouping method was applied to PBC dataset, and three subgroups have been detected. Accordingly, there are also significant differences in the severity of the disease among the three subgroups. This result has shown the usefulness and superiority of our method from a practical point of view.

As continuous response has been considered in this work, one potential future work is to extend the proposed framework to discrete longitudinal outcomes. Second, when the dimension of covariates is high, it would cost more time to adopt the proposed method. So another open issue for future research is extending the framework to high dimensional additive model \citep{panagiotelis2008bayesian,meier2009high,fan2011nonparametric}. In addition, dealing with the dataset where interactions between covariates exist is also a possible direction. These issues are challenging but deserve further exploration.

%\newpage
\bibliographystyle{apalike}
\bibliography{ref}
%\clearpage

\section*{Appendix}

\begin{lemma}
\label{lem1}
Without loss of generality, we consider the following B-spline basis functions that span $G$, that is, $B_k= J_n^{1/2} S_k,k=1,\ldots, J_n$, where $\{S_k\}_{k=1}^{J_n}$ are the B-splines defined in Chapter 5 of \cite{lorentz1993constructive}. It follows from Theorem 4.2 f \cite{lorentz1993constructive} that
$$M_1 \| \gamma\|_2^2 \leq \int \{ \sum_{k=1}^{J_n} B_k(t) \gamma_k \}^2 dt \leq M_2 \| \gamma \|_2^2$$
for some constant $0< M_1 < M_2 <\infty$, where $\gamma = (\gamma_1,\ldots,r_{J_n})^{\prime}$.
\end{lemma}

\begin{lemma}
\label{lem2}
For each $i$, there exist some constants $0<M_1<M_2<\infty$ such that, except on an event whose probability tends to zero, all the eigenvalues of $\bm X_i^{\prime} \bm X_i/n_i$ fall between $M_1$ and $M_2$.
\end{lemma}
\noindent The proof of Lemma 2 can be referred to Theorem 7.3 of \cite{roussas1987moment}.

\begin{lemma}
\label{lem3}
Define $f_i^{*}(t)=\bm B(t) \gamma^{*} \in G$ such that $\|f_i^{*}(t)-f_i(t)\|_2^2=inf_{g \in G} \|g(t) - f_i(t)\|_2^2 \triangleq \varpi_{i}=J_n^{-2r}$ if $f_i \in \mathcal H_r$, which follows the result of page 149 of \cite{de1978practical}.
\end{lemma}

\begin{lemma}
\label{lem4}
Under Conditions (A1)-(A5), as $n \rightarrow \infty$ and if $J_n\rightarrow \infty, J_n^4=o(n_0)$, then the B-splines coefficients we import into $k$-means satisfy
$$ \| \widehat \gamma_{ij}^{(1)} - \gamma_{ij} \|_2^2 = O_p(J_n/n_i + J_n^{-2r}),$$
where $i=1,\ldots,n,j=1,\ldots,p$, $\gamma_{ij}$ denotes the true B-splines coefficients for the $j$ additive component of the $i$ subject, and $\widehat \gamma_{ij}^{(r)}$ represents the estimate of $\gamma_{ij}$ at the first iteration.
\end{lemma}

\begin{remark}
In Lemma 4, we investigate the convergence property of our initial estimator, which means that in the first iteration, the B-spline coefficients we input into $k$-means algorithm is close to the true coefficients as long as the repeated measurements of each subjects are sufficiently large, and this is consistent with our simulation studies. Moreover, we have noticed that the average mean squared error of our estimated coefficients is determined by two parts. The first part consists of average variance and approximation bias from the random effects, and the second part represents average squared shrinkage bias. Lemma 4 guarantees that the estimated subgroups $\widehat \Omega_{k,j}$ equal to the true subgroups $\Omega_{k,j}$ with probability $1$.
\end{remark}

\begin{proof}[Proof of Lemma \ref{lem4}]
Denote $\bm U_i = (\bm S_i, \bm X_i)$, following the idea of \cite{huang2007efficient}, the initial estimates are
$$\left(\begin{array}{l}
\widehat{\beta} \\
\widehat{\gamma}
\end{array}\right) = (\sum_{i=1}^{n} \bm U_i^{\prime} \bm V_i^{-1} \bm U_i)^{-1} \sum_{i=1}^{n} \bm U_i^{\prime} \bm V_i^{-1} \bm Y_i.$$
Let
$$
\sum_{i=1}^{n} {\bm U}_{i}^{\prime} \bm V_{i}^{-1} {\bm U}_{i}=\left(\begin{array}{ll}
\sum_{i=1}^{n} {\bm S}_{i}^{\prime} \bm V_{i}^{-1} {\bm S}_{i} & \sum_{i=1}^{n} {\bm S}_{i}^{\prime} \bm V_{i}^{-1} {\bm X}_{i} \\
\sum_{i=1}^{n} {\bm X}_{i}^{\prime} \bm V_{i}^{-1} {\bm S}_{i} & \sum_{i=1}^{n} {\bm X}_{i}^{\prime} \bm V_{i}^{-1} {\bm S}_{i}
\end{array}\right) \triangleq\left(\begin{array}{ll}
H_{11} & H_{12} \\
H_{21} & H_{22}
\end{array}\right),
$$
and it follows that well-known block matrix forms of matrix inverse that
$$
\left(\begin{array}{cc}
H_{11} & H_{12} \\
H_{21} & H_{22}
\end{array}\right)^{-1}=\left(\begin{array}{cc}
H^{11} & H^{12} \\
H^{21} & H^{22}
\end{array}\right)=\left(\begin{array}{cc}
H_{11 \cdot 2}^{-1} & -H_{11 \cdot 2}^{-1} H_{12} H_{22}^{-1} \\
-H_{22 \cdot 1}^{-1} H_{21} H_{11}^{-1} & H_{22 \cdot 1}^{-1}
\end{array}\right),
$$
where $H_{11 \cdot 2}=H_{12} H_22^{-1} H_{21}$ and $H_{11 \cdot 2} = H_{22} -H_{21} H_{11}^{-1} H_{12}.$ Consequently,
$$
\widehat{\beta}=H^{11}\left\{\sum_{i=1}^{n} {\bm S}_{i}^{\prime} \bm V_{i}^{-1} {\bm Y}_{i}-H_{12} H_{22}^{-1} \sum_{i=1}^{n} {\bm X}_{i}^{\prime} \bm V_{i}^{-1} {\bm Y}_{i}\right\},
$$
$$
\widehat \gamma^{(0)} = H^{22}\left\{ \sum_{i=1}^{n} \bm X_i^{\prime} \bm V_i^{-1} \bm Y_i - H_{21} H_{11}^{-1} \sum_{i=1}^{n} \bm S_i \bm V_i^{-1} \bm Y_i \right \},
$$
where $\widehat \gamma^{(0)}= (\hat \gamma_1^{(0)},\ldots, \widehat \gamma_p^{(0)})^{\prime}$.
Next we are going to prove that the initial estimate of B-splines coefficients are bounded with probability and let us take $j=1$ as example in the following proof.
In the backfitting procedure, we define
$$\bm W_{i1}=\bm Y_i - \bm S_i \widehat \beta - \sum_{k=2}^{p} \bm B_{ik} \hat \gamma_k^{(0)},$$
where $\bm X_{ik}=(X_{i1k},\ldots, X_{i{n_i}k})$, and update $ f_{i1} = S(\bm W_{i1} | x_{i1})$ with B-splines, which means in this situation we fit
\begin{align*}
	\bm W_{i1} & = f_{i1}(\bm X_{i1}) + \bm Z_i b_i + \bm \varepsilon_i  \\
           & \approx \bm B_{i1} r_{i1} + \bm Z_i b_i + \bm \varepsilon_i.
\end{align*}
So the initial estimates we import into k-means algorithm are
$$\widehat \gamma_{i1}^{(1)} = (\bm B_{i1}^{\prime} \bm V_i^{-1} \bm B_{i1} )^{-1} \bm B_{i1}^{\prime} \bm V_i^{-1} \bm W_{i1}.$$
Thus,
\begin{align*}
	& \|\widehat \gamma_{i1}^{(1)}- \gamma_{i1}^{*}\|^2_2 = \|(\bm B_{i1}^{\prime} \bm V_i^{-1} \bm B_{i1} )^{-1} \bm B_{i1}^{\prime} \bm V_i^{-1} \bm W_{i1}- \gamma_{i1}^{*}\|_2^2 \\
	& = \| (\bm B_{i1}^{\prime} \bm V_i^{-1} \bm B_{i1} )^{-1} \bm B_{i1}^{\prime} \bm V_i^{-1} (\bm Y_i - \bm S_i \widehat \beta - \sum_{k=2}^{p}\bm B_{ik} \widehat \gamma_k^{(0)}) -\gamma_{i1}^{*} \|_2^2 \\
	& = \| (\bm B_{i1}^{\prime} \bm V_i^{-1} \bm B_{i1} )^{-1} \bm B_{i1}^{\prime} \bm V_i^{-1} (\bm S_i(\beta-\widehat \beta) + f_{i1} + \sum_{k=2}^{p} (f_{ik} - \bm B_{ik}\widehat \gamma_k^{(0)}) + \bm Z_i b_i +\varepsilon_i) -\gamma_{i1}^{*}\|_2^2 \\
	&  \leq \| (\bm B_{i1}^{\prime} \bm V_i^{-1} \bm B_{i1} )^{-1} \bm B_{i1}^{\prime} \bm V_i^{-1} \bm S_i (\beta- \widehat \beta)\|^2_2+\|(\bm B_{i1}^{\prime} \bm V_i^{-1} \bm B_{i1} )^{-1} \bm B_{i1}^{\prime} \bm V_i^{-1}f_{i1} -\gamma_{i1}^{*} \|_2^2 \\
	& \quad +  \|  (\bm B_{i1}^{\prime} \bm V_i^{-1} \bm B_{i1} )^{-1} \bm B_{i1}^{\prime} V_i^{-1} \sum_{k=2}^{p} (f_{ik} - \bm B_{ik}\widehat \gamma_k^{(0)})\|_2^2 + \|(\bm B_{i1}^{\prime} \bm V_i^{-1} \bm B_{i1} )^{-1} \bm B_{i1}^{\prime} \bm V_i^{-1} ( \bm Z_i b_i +\bm \varepsilon_i)\|_2^2 \\
	& \triangleq I_1 + I_2 + I_3 + I_4.
\end{align*}
For $I_1$, it is obviously that
\begin{align*}
	& E(I_1) =  E(\|(\bm B_{i1}^{\prime} \bm V_i^{-1} \bm B_{i1} )^{-1} \bm B_{i1}^{\prime} \bm V_i^{-1} S_i (\beta- \widehat \beta) \|_2^2) \\
	& = E((\beta-\widehat \beta)^{\prime} \bm S_i^{\prime} \bm V_i^{-1} \bm B_{i1} (\bm B_{i1}^{\prime} \bm V_i^{-1} \bm B_{i1} )^{-1} (\bm B_{i1}^{\prime} \bm V_i^{-1} \bm B_{i1} )^{-1} \bm B_{i1}^{\prime} \bm V_{i}^{-1} \bm S_i (\beta - \widehat \beta) )\\
	& = tr\{ E\{(\beta-\widehat \beta)^{\prime} \bm S_i^{\prime} \bm V_i^{-1} \bm B_{i1} (\bm B_{i1}^{\prime} \bm V_i^{-1} \bm B_{i1} )^{-1} (\bm B_{i1}^{\prime} \bm V_i^{-1} \bm B_{i1} )^{-1} \bm B_{i1}^{\prime} \bm V_{i}^{-1} \bm S_i (\beta - \widehat \beta) \} \} \\
	%& = E\{ trace\{ (\beta-\widehat \beta)^{\prime} S_i^{\prime} V_i^{-1} B_{i1} (B_{i1}^{\prime} V_i^{-1} B_{i1} )^{-1} (B_{i1}^{\prime} V_i^{-1} B_{i1} )^{-1} B_{i1}^{\prime} V_{i}^{-1} S_i (\beta - \widehat \beta) ) \}\} \\
	%& = E\{ trace\{ S_i (\beta - \widehat \beta) (\beta - \widehat \beta)^{\prime} S_i^{\prime} V_i^{-1} B_{i1} (B_{i1}^{\prime} V_i^{-1} B_{i1} )^{-1} (B_{i1}^{\prime} V_i^{-1} B_{i1} )^{-1} B_{i1}^{\prime} V_{i}^{-1} \}\} \\
	& = tr\{ E\{ \bm S_i (\beta - \widehat \beta) (\beta - \widehat \beta)^{\prime} \bm S_i^{\prime} \bm V_i^{-1} \bm B_{i1} (\bm B_{i1}^{\prime} \bm V_i^{-1} \bm B_{i1} )^{-1} (\bm B_{i1}^{\prime} \bm V_i^{-1} \bm B_{i1} )^{-1} \bm B_{i1}^{\prime} \bm V_{i}^{-1} \} \} \\
	& \leq \lambda_{max}\{\bm S_i E((\beta - \widehat \beta) (\beta - \widehat \beta)^{\prime} )\bm S_i^{\prime}\}
	 tr\{ \bm V_i^{-1} \bm B_{i1} (\bm B_{i1}^{\prime} \bm V_i^{-1} \bm B_{i1} )^{-1} (\bm B_{i1}^{\prime} \bm V_i^{-1} \bm B_{i1} )^{-1} \bm B_{i1}^{\prime} \bm V_{i}^{-1} \} \\
	%& \leq \lambda_{max}\{S_i E((\beta - \widehat \beta) (\beta - \widehat \beta)^{\prime} )S_i^{\prime}\} \lambda_{max}(B_{i1}^{\prime} V_i^{-1} V_i^{-1} B_{ii1}) trace\{ (B_{i1}^{\prime} V_i^{-1} B_{i1} )^{-1} (B_{i1}^{\prime} V_i^{-1} B_{i1} )^{-1} \} \\
	%& \leq \lambda_{max}\{S_i E((\beta - \widehat \beta) (\beta - \widehat \beta)^{\prime} )S_i^{\prime}\} \lambda_{max}(B_{i1}^{\prime} V_i^{-1} V_i^{-1} B_{i1}) \lambda_{max}((B_{i1}^{\prime} V_i^{-1} B_{i1} )^{-1}) trace\{(B_{i1}^{\prime} V_i^{-1} B_{i1} )^{-1}\} \\
	%& = \lambda_{max}\{S_i E((\beta - \widehat \beta) (\beta - \widehat \beta)^{\prime} )S_i^{\prime}\} \lambda_{max}(B_{i1}^{\prime} V_i^{-1} V_i^{-1} B_{i1}) \lambda_{max}((B_{i1}^{\prime} V_i^{-1} B_{i1} )^{-1}) tr\{ n_i^{-1} F_{i1}^{-1}\} \\
	& =O(J_n/n_i).
\end{align*}
Following Lemma A1 of \cite{zhu2008asymptotics}, since $\lambda_{max}\{\bm S_i E((\beta - \widehat \beta) (\beta - \widehat \beta)^{\prime} )\bm S_i^{\prime}\} < C_1$, $\lambda_{max}(\bm B_{i1}^{\prime} \bm V_i^{-1} \bm V_i^{-1} \bm B_{i1}) <C_2$, and $\lambda_{max}((\bm B_{i1}^{\prime} \bm V_i^{-1} \bm B_{i1} )^{-1})<C_3$, by defining $F_{i1} = \frac{1}{n_i} \bm B_{i1}^{\prime} \bm V_i^{-1} \bm B_{i1}$, we have $\| F_{i1}^{-1}\|_{\infty} = O(J_n)$. According to Lemma 6.5 of \cite{shen1998local}, we obtain the above inequalities, thus $I_1 = O_p(J_n/n_i)$.\\
For $I_2$, we have
\begin{align*}
I_2 &= 	\|(\bm B_{i1}^{\prime} \bm V_i^{-1} \bm B_{i1} )^{-1} \bm B_{i1}^{\prime} \bm V_i^{-1}f_{i1} -\gamma_{i1} \|_2^2 \\
& = \| (\bm B_{i1}^{\prime} \bm V_i^{-1} \bm B_{i1} )^{-1} \bm B_{i1}^{\prime} \bm V_i^{-1} (f_{i1}^{*} +f_{i1} - f_{i1}^{*}) - \gamma_{i1}^{*} )\|_2^2 \\
& = \| (\bm B_{i1}^{\prime} \bm V_i^{-1} \bm B_{i1} )^{-1} \bm B_{i1}^{\prime} \bm V_i^{-1} (f_{i1} - f_{i1}^{*})\|_2^2 \\
& = tr\{ (f_{i1} - f_{i1}^{*})^{\prime}  \bm V_i^{-1} \bm B_{i1} (\bm B_{i1}^{\prime} \bm V_i^{-1} \bm B_{i1} )^{-1} (\bm B_{i1}^{\prime} \bm V_i^{-1} \bm B_{i1} )^{-1} \bm B_{i1}^{\prime} \bm V_i^{-1} (f_{i1} - f_{i1}^{*})\} \\
%& = trace\{ V_i^{-1} B_{i1} (B_{i1}^{\prime} V_i^{-1} B_{i1} )^{-1} (B_{i1}^{\prime} V_i^{-1} B_{i1} )^{-1} B_{i1}^{\prime} V_i^{-1} (f_{i1} - f_{i1}^{*}) (f_{i1} - f_{i1}^{*})^{\prime}\} \\
& \leq \lambda_{max} (\bm V_i^{-1} \bm B_{i1} (\bm B_{i1}^{\prime} \bm V_i^{-1} \bm B_{i1} )^{-1} (\bm B_{i1}^{\prime} \bm V_i^{-1} \bm B_{i1} )^{-1} \bm B_{i1}^{\prime} \bm V_i^{-1}) trace\{ (f_{i1} - f_{i1}^{*}) (f_{i1} - f_{i1}^{*})^{\prime} \} \\
& = \lambda_{max} (\bm V_i^{-1} \bm B_{i1} (\bm B_{i1}^{\prime} \bm V_i^{-1} \bm B_{i1} )^{-1} (\bm B_{i1}^{\prime} \bm V_i^{-1} \bm B_{i1} )^{-1} \bm B_{i1}^{\prime} \bm V_i^{-1}) \| f_{i1} - f_{i1}^{*}  \|_2^2.
\end{align*}
Since $\lambda_{max} (\bm V_i^{-1} \bm B_{i1} (\bm B_{i1}^{\prime} \bm V_i^{-1} \bm B_{i1} )^{-1} (\bm B_{i1}^{\prime} \bm V_i^{-1} \bm B_{i1} )^{-1} \bm B_{i1}^{\prime} \bm V_i^{-1}) < C$,following from the result of \cite{de1978practical}, we know that $\| f_{i1} - f_{i1}^{*}  \|_2^2= J_n^{-2r}$, and then $I_2 = O_p(J_n^{-2r})$. \\
Next, for $I_3$,
\begin{align*}
	&E(I_3)  = E\{ \| (\bm B_{i1}^{\prime} \bm V_i^{-1} \bm B_{i1} )^{-1} \bm B_{i1}^{\prime} \bm V_i^{-1}  \sum_{k=2}^{p} (f_{ik} - \bm B_{ik}\widehat \gamma_k^{(0)}) \|_2^2\} \\
	& = tr\{ E\{ \sum_{k=2}^{p} (f_{ik} - \bm B_{ik}\widehat \gamma_k^{(0)})^{\prime} \bm V_i^{-1} \bm B_{i1} (\bm B_{i1}^{\prime} \bm V_i^{-1} \bm B_{i1} )^{-1} (\bm B_{i1}^{\prime} \bm V_i^{-1} \bm B_{i1} )^{-1} \bm B_{i1}^{\prime} \bm V_i^{-1} \sum_{k=2}^{p} (f_{ik} - \bm B_{ik}\widehat \gamma_k^{(0)}) \}\} \\
	%& = E\{ trace\{ \sum_{k=2}^{p} (f_{ik} - B_{ik}\widehat \gamma_k^{(0)})^{\prime} V_i^{-1} B_{i1} (B_{i1}^{\prime} V_i^{-1} B_{i1} )^{-1} (B_{i1}^{\prime} V_i^{-1} B_{i1} )^{-1} B_{i1}^{\prime} V_i^{-1} \sum_{k=2}^{p} (f_{ik} - B_{ik}\widehat \gamma_k^{(0)})\}\} \\
	%& = E\{trace\{ \sum_{k=2}^{p} (f_{ik} - B_{ik}\widehat \gamma_k^{(0)}) \sum_{k=2}^{p} (f_{ik} - B_{ik}\widehat \gamma_k^{(0)})^{\prime} V_i^{-1} B_{i1} (B_{i1}^{\prime} V_i^{-1} B_{i1} )^{-1} (B_{i1}^{\prime} V_i^{-1} B_{i1} )^{-1} B_{i1}^{\prime} V_i^{-1} \}\} \\
	& = tr\{ E\{ \sum_{k=2}^{p} (f_{ik} - \bm B_{ik}\widehat \gamma_k^{(0)}) \sum_{k=2}^{p} (f_{ik} - \bm B_{ik}\widehat \gamma_k^{(0)})^{\prime}\}  \bm V_i^{-1} \bm B_{i1} (\bm B_{i1}^{\prime} \bm V_i^{-1} \bm B_{i1} )^{-1} (\bm B_{i1}^{\prime} \bm V_i^{-1} \bm B_{i1} )^{-1} \bm B_{i1}^{\prime} \bm V_i^{-1} \} \\
	& \leq \lambda_{max}( \sum_{k=2}^{p} (f_{ik} - \bm B_{ik}\widehat \gamma_k^{(0)}) \sum_{k=2}^{p} (f_{ik} - \bm B_{ik}\widehat \gamma_k^{(0)})^{\prime}) \\
	& \quad \times tr\{ \bm V_i^{-1} \bm B_{i1} (\bm B_{i1}^{\prime} \bm V_i^{-1} \bm B_{i1} )^{-1} (\bm B_{i1}^{\prime} \bm V_i^{-1} \bm B_{i1} )^{-1} \bm B_{i1}^{\prime} \bm V_i^{-1}\} \\
	& = \lambda_{max}( \sum_{k=2}^{p} (f_{ik} - \bm B_{ik}\widehat \gamma_k^{(0)}) \sum_{k=2}^{p} (f_{ik} - \bm B_{ik}\widehat \gamma_k^{(0)})^{\prime}) \lambda_{max}( \bm B_{i1}^{\prime} \bm V_i^{-1} \bm V_i^{-1} \bm B_{i1} ) \\
	& \quad \times \lambda_{max}((\bm B_{i1}^{\prime} \bm V_i^{-1} \bm B_{i1} )^{-1}) tr\{n_i^{-1} F_i^{-1}\} \\
	& = O(J_n/n_i)
\end{align*}
Similarly to the proof of $I_1$, since $\| F_{i1}^{-1}\|_{\infty}=O(J_n)$ and $\lambda_{max}( \sum_{k=2}^{p} (f_{ik} - \bm B_{ik}\widehat \gamma_k^{(0)}) \sum_{k=2}^{p} (f_{ik} - \bm B_{ik}\widehat \gamma_k^{(0)})^{\prime}) ,\lambda_{max}( \bm B_{i1}^{\prime} \bm V_i^{-1} \bm V_i^{-1} \bm B_{i1} ),\lambda_{max}((\bm B_{i1}^{\prime} \bm V_i^{-1} \bm B_{i1} )^{-1})$ are bounded away from zero to infinity, we can obtain $I_3=O_p(J_n/n_i)$.

Furthermore, for $I_4$, by Lemma \ref{lem4} and the bounded assumption on the eigenvalues of $V$, it is easy to verify that there exist two constants $0<C_1 \leq C_2 < \infty$, such that
$$
  C_1 \frac{1}{ n_i^2} E((\bm Z_i b_i +\bm \varepsilon_i)^{\prime} \bm V_i^{-1} \bm B_{i1} \bm B_{i1}^{\prime} \bm V_i^{-1} (\bm Z_i b_i +\bm \varepsilon_i)) \leq   E(I_4),$$
 $$  E(I_4) \leq  C_2 \frac{1}{ n_i^2} E((\bm Z_i b_i +\bm \varepsilon_i)^{\prime} \bm V_i^{-1} \bm B_{i1} \bm B_{i1}^{\prime} \bm V_i^{-1} (\bm Z_i b_i +\bm \varepsilon_i)).$$
According to the operation properties of the trace and expectation, we have
\begin{align*}
& E((\bm Z_i b_i +\bm \varepsilon_i)^{\prime} \bm V_i^{-1} \bm B_{i1} \bm B_{i1}^{\prime} \bm V_i^{-1} (\bm Z_i b_i +\varepsilon_i)) = tr\{ E((\bm Z_i b_i +\bm \varepsilon_i)^{\prime} \bm V_i^{-1} \bm B_{i1} \bm B_{i1}^{\prime} \bm V_i^{-1} (\bm Z_i b_i +\varepsilon_i))\}	\\
&= E(tr\{ \bm B_{i1}^{\prime} \bm V_i^{-1} (\bm Z_i b_i +\varepsilon_i) (\bm Z_i b_i +\bm \varepsilon_i)^{\prime} \bm V_i^{-1} \bm B_{i1} \})= E(tr\{ \bm B_{i1}^{\prime} \bm V_i^{-1} \bm \Sigma_{i} \bm V_i^{-1} \bm B_{i1}\}) \\
& = tr\{ E(\bm B_{i1}^{\prime} \bm V_i^{-1} \bm \Sigma_{i} \bm V_i^{-1} \bm B_{i1})\} = O_p(n_i J_n).
\end{align*}
Hence, we obtain $I_4= O_p(J_n/n_i)$.\\
Consequently, combining the $I_1,I_2,I_3,I_4$, we have
$$\| \widehat \gamma_{i1}^{(1)} - \gamma_{i1}\|_2^2 \leq \| \widehat \gamma_{i1}^{(1)} - \gamma_{i1}^{*}\|_2^2 +\|\gamma_{i1}^{*}-\gamma_{i1}\|_2^2=O_p(J_n/n_i + J_n^{-2r}).$$
This finishes the proof.
\end{proof}

\begin{lemma}
\label{lem5}
If $log c/log n \rightarrow 0$ as $n_0 \rightarrow \infty$, then we have $$\hat m_j \stackrel{P}{\rightarrow} m_j,$$ for $j=1,\ldots,p$, where $\hat m_j$ is the number of groups selected by BIC, $m_j$ is the true number of groups on covariate $X_j$, and $c= \mathop{max} \limits_{\substack{j=1,\ldots,p\\ k=1,\ldots, \hat m_j}}|\Omega_{k,j}|$.
\end{lemma}

\begin{remark}
Lemma 5 implies that when the true number of groups $m_j$ is unknown, BIC can be utilized to determine the number of clusters and the estimated number of groups $\hat m_j$ selected by BIC converges to $m_j$ in probability. This is consistent with some findings for BIC in tuning parameter determinations, for instance, \cite{zhang2010regularization} and \cite{bai2018consistency}.
\end{remark}

\begin{proof}[Proof of Lemma \ref{lem5}]
If $\hat m_j < m_j$, then we can find at least one pair of $\gamma_i$ and $\gamma_{i ^ \prime}$, such that they are not in the same cluster but they are grouped to the same cluster.
By the Lemma 1 of \cite{zhang2021generalized}, for $$BIC(\hat m_j)= -2 log(\hat L) + log(n) \times \hat m_j \times df,$$
where $df$ represents the degree of freedom of B-splines, the first item on the right-hand side goes to $\infty$ with rate $n_0$. It is faster than the rate of BIC under $\hat m_j = m_j$, implying that $P(\hat m_j < m_j)=0$ as $n_0 \rightarrow \infty$. Therefore, we only need to prove that $\hat m_j > m_j$ leads to $BIC(\hat m_j) > BIC( m_j)$.
Most of time, the degree of freedom of B-splines is not less than 5, and it suffices to show that for
\begin{align*}
P(BIC(m_j)<BIC(\hat m_j))  = P(-2(\hat l(m_j) - \hat l(\hat m_j))<log(n) \times df \times (\hat m_j- m_j))
\end{align*}
according to \cite{donoho2004higher}, the limiting distribution of $-2(\hat l(m_j) - \hat l(\hat m_j))$ is $\chi^{2}$ distribution. The right-hand side of above formula is tending to 1 as $n \rightarrow \infty$. Hence we conclude that $P(BIC(m_j)<BIC(\hat m_j)) \rightarrow 1$ when $\hat m_j > m_j$, indicating that $\hat{m}_{j} \stackrel{P}{\rightarrow} m_{j}$. This finishes the proof.
\end{proof}

\begin{proof}[Proof of Theorem \ref{th1}]
When the true subgroup memberships $\Omega_{1,j},\ldots,\Omega_{m_j,j},j=1,\ldots,p$ are known, the oracle approximation for $\tilde \gamma_j$ can be easily obtained.
Taking $j=1$ as an example, for $i \in \Omega_{1,j}$, we have
$$\tilde \gamma_{11} = (\bm B_1^{\prime} \bm V_1^{-1} \bm B_1)^{-1} \bm B_1^{\prime} \bm V_1^{-1} (\bm Y_1 - \bm S_1 \widehat \beta - \sum_{j=2}^{p} \bm B_j \gamma_j^{(0)}),$$
where $\bm B_1 = (B_{11}^{\prime},\ldots,B_{N_{1,1}1}^{\prime})^{\prime}$ is the B-spline design matrix, $\bm V_1= diag(V_{11},\ldots,V_{N_{1,1}1})$ is the working correlation matrix of whom belonging to $\Omega_{1,1}$.

Similarly to the proof of Lemma \ref{lem4}, we can show that under Assumptions (A1)-(A5),
$$\| \tilde \gamma_j - \gamma_j\|_2^2 \leq O_p( J_n/N_{0}+ J_n^{-2r}),j=1,\ldots,p.$$
This finishes the proof.
\end{proof}

\begin{proof}[Proof of Theorem \ref{th2}]
When the algorithm converges, which means when $\widehat \Omega_j^{(n)} = \widehat \Omega_j^{(n-1)}$ for $\forall j=1,\ldots,p$, the proposed method stops. Assuming that for covariate $X_j$, there are $\hat m_j$ groups being detected, $\widehat \Omega_{1,j},\ldots,\widehat \Omega_{\hat m_j,j}$. And the number of membership of each subgroup are $n_{1,j},\ldots,n_{\hat m_j,j}$ respectively, where $\sum_{k=1}^{\hat m_j} n_{k,j} = n$.\\
Let $\bm B_{1,j},\ldots, \bm B_{\hat m_j,j}$ denote the B-spline matrix of the $\widehat \Omega_{1,j},\ldots, \widehat \Omega_{\hat m_j,j}$ group respectively, and $\bm S_{1,j},\ldots,\bm S_{\hat m_j,j}$ denote the design matrix for the individuals of group $\widehat \Omega_{1,j},\ldots, \widehat \Omega_{\hat m_j,j}$. \\
So $\widehat \gamma_1.\ldots, \widehat \gamma_p$, the estimates of $\gamma_1,\ldots,\gamma_p$ after convergence are obtained as
$$\widehat \gamma_{k,j}= (\bm B_{k,j}^{\prime} \bm V_{k,j}^{-1} \bm B_{k,j})^{-1} \bm B_{k,j}^{\prime} \bm V_{k,j}^{-1} (y_{k,j} -\bm S_{k,j} \widehat \beta - \sum_{l \neq j} \bm B_{k,l} \widehat \gamma_{k,l}),\quad k=1,\ldots,\hat m_j,j=1,\ldots,p$$
Similarly to the proof the Lemma \ref{lem4}, we can easily conclude that
$$\| \widehat \gamma_{k,j} -\gamma_{k,j} \|_2^2 = O_p(J_n/n_{0k} +J_n^{-2r}),$$
where $n_{0k} = min \{n_i, i \in \widehat \Omega_{k,j} \}$.
As a result, we have $$\| \widehat \gamma_{j} -\gamma_{j} \|_2^2 = O_p(J_n/n_{0} +J_n^{-2r}),$$
which leads to $$\| \widehat f_j - f_j \|_2^2 = O_p(J_n/n_0 +J_n^{-2r}).$$
This finishes the proof.
\end{proof}

\begin{proof}[Proof of Theorem \ref{th3}]
According to Theorem 1 of \cite{abraham2003unsupervised}, if $n$ is sufficiently large, as long as $\| \widehat \gamma_{ij}^{(1)} - \gamma_{ij}\|_2^2 \rightarrow 0$, for the K-means procedure, the unique minimizer of objective function $$\mu_n = \sum_{j=1}^{k} \sum_{i \in G_j} \| \widehat \gamma_i - c_j\|_2^2$$ is arbitrarily close to the unique minimizer of objective function
$$\mu = \sum_{j=1}^{k} \sum_{i \in G_j} \| \gamma_i - c_j\|_2^2.$$
With Theorem 1 of \cite{abraham2003unsupervised} ensuring the consistency of K-means, we can establish that as $n \rightarrow \infty$ and $\| \widehat \gamma_{ij}^{(\infty)} - \gamma_{ij}\|_2^2 \rightarrow 0$, the proposed method can identify the true subgroup structure with probability tending to 1. Hence, we have $$P(\widehat \Omega_{j} = \Omega_{j}) \rightarrow 1.$$
This finishes the proof.
\end{proof}

\end{document}